\documentclass[english]{article}
\usepackage[utf8]{inputenc} 
\usepackage[english]{babel} 
\usepackage[noend]{algpseudocode}
\usepackage[hyphens]{url}
\usepackage[pdfborder={0 0 0}]{hyperref}
\usepackage{cite}
\usepackage{graphicx} 
\usepackage{caption}
\usepackage{subcaption}
\usepackage{amsmath,amssymb}
\usepackage{float}
\usepackage{todonotes}
\usepackage[slant]{complexity}
\usepackage{thm-restate}
\usepackage{xspace}
\usepackage{microtype}
\usepackage{pdfpages}
\usepackage{color}
\usepackage{amsthm}
\usepackage{authblk}
\usepackage{fullpage}
\newtheorem{definition}{Definition}
\newtheorem{theorem}[definition]{Theorem}
\newtheorem{lemma}[definition]{Lemma}
\newtheorem{corollary}[definition]{Corollary}
\newtheorem{observation}[definition]{\textbf{Observation}}

\newcommand{\CHull}{\ensuremath{Conv}}

\newcommand{\new}[1]{\color{black}#1\color{black}\xspace}
\newcommand{\nnew}[1]{\color{black}#1\color{black}\xspace}
\newcommand{\statement}[2]{\vskip2ex\noindent{\sffamily\bfseries#1~\ref{#2}.}}
\newcommand{\remove}[1]{}
\date{}

\title{Coordinated Motion Planning: Reconfiguring a Swarm of Labeled Robots with Bounded Stretch\footnote{This work was partially supported by the DFG Research Unit {\em Controlling Concurrent Change}, funding number FOR 1800, project FE407/17-2, Conflict Resolution and Optimization.}}

\author[1]{Erik D. Demaine}
\author[2]{S\'{a}ndor P. Fekete}
\author[2]{Phillip Keldenich}
\author[3]{Henk~Meijer}
\author[2]{Christian Scheffer}
\affil[1]{MIT Computer Science and Artificial Intelligence Laboratory, Cambridge, MA, USA\\
  \texttt{edemaine@mit.edu}}
\affil[2]{Department of Computer Science, TU Braunschweig Braunschweig, Germany\\
  \texttt{$\{$s.fekete,p.keldenich,c.scheffer$\}$@tu-bs.de}}
\affil[3]{Science Department, University College Roosevelt Middelburg, The Netherlands\\
  \texttt{h.meijer@ucr.nl}}

\newcommand{\revised}[1]{{\color{black} #1}}

\begin{document}
	\maketitle
	\begin{abstract}
		We present a number of breakthroughs for coordinated motion planning, in which the
objective is to reconfigure a swarm of labeled convex objects by
a combination of parallel, continuous, collision-free translations into a given
target arrangement.  Problems of this type can be traced back to the classic work
of Schwartz and Sharir (1983), who gave a method for deciding the existence 
of a coordinated motion for a set of disks between obstacles; their approach is polynomial in the 
complexity of the obstacles, but exponential in the number of disks.
Other previous work has largely focused on {\em sequential} schedules, in which one robot moves at a time.

We provide
constant-factor approximation algorithms for minimizing the
execution time of a coordinated, {\em parallel} motion plan for a swarm of robots in the absence of obstacles,
provided some amount of separability.

Our algorithm achieves {\em constant stretch factor}: If
all robots are at most $d$ units from their respective
starting positions, the total duration of the overall schedule
is $O(d)$.
Extensions include unlabeled robots and different classes of robots. We
also prove that finding a plan with minimal execution time is
NP-hard, even for a grid arrangement without any stationary obstacles. 
On the other hand, we show that for densely packed disks that cannot be well separated, a stretch factor $\Omega(N^{1/4})$
may be required. On the
positive side, we establish a stretch factor of $O(N^{1/2})$ even in this case.
	\end{abstract}

	\section{Introduction}\label{sec:intro}
Since the beginning of computational geometry, robot motion planning has
been at the focus of algorithmic research. Planning the relocation of a geometric
object among geometric obstacles leads to intricate scientific challenges,
requiring the combination of deep geometric and mathematical insights with
algorithmic techniques.
With the broad and ongoing progress in robotics, the increasing importance of 
intelligent global planning with performance guarantees
requires more sophisticated algorithmic reasoning,
in particular when it comes to the higher-level task of coordinating the motion of many robots.

From the early days, multi-robot coordination has received attention from the algorithmic side.
Even in the groundbreaking work by Schwartz and Sharir~\cite{ss-pmpcbpb-83} from the 1980s, one of the
challenges was coordinating the motion of {\em several} disk-shaped objects
among obstacles. Their algorithms run in time polynomial in the complexity of the obstacles, but 
exponential in the number of disks. 
This illustrates the significant challenge of coordinating many individual robots.
In addition, a growing number of applications focus primarily on robot interaction in the absence of obstacles,
such as air traffic control or swarm robotics,
where the goal is overall efficiency, rather than individual navigation.

With the challenges of multi-robot coordination being well known, there is still a huge demand
for positive results with provable performance guarantees. In this paper, we provide
significant progress in this direction, with a broad spectrum of results.

\subsection{Our Results}
\begin{itemize}
\item
	For the problem of minimizing the total time for reconfiguring a system of labeled
circular robots in a grid environment, we show that it is strongly NP-complete to compute an optimal solution; see Theorem~\ref{thm:optimalNPcomplete}. 
\item We give an $\mathcal{O}(1)$-approximation for
the long-standing open problems of parallel motion-planning with minimum makespan in a grid setting.
This result is based on establishing an {\em absolute} performance guarantee: We prove that for any labeled arrangement of robots, 
there is always an overall schedule that gets each robot to its target destination
with bounded {\em stretch}, i.e., within a constant factor of the largest individual distance. 
See Theorem~\ref{thm:main} for the base case of grid-based configurations, which is extended later on.
\item
For our approach, we make use of a technique to separate planar (cyclic) flows
into so-called {\em subflows} whose thickness can be controlled by the number of
subflows, see Definition~\ref{def:partition} and
Lemma~\ref{lem:computePartition}. This is of independent interest for the area	
of packet routing with bounded memory: Our Theorem~\ref{thm:main2} implies that 
$\mathcal{O}(D)$ steps are sufficient to route any permutation of dilation $D$ on the grid, 
even with a buffer size of 1, resolving an open question by Scheideler~\cite{s-ursin-98} 
dating back to 1998.


\item We extend our approach to establish constant stretch 
for the generalization of {\em colored} robot classes, for which unlabeled robots are another special case; see Theorem~\ref{thm:unlabeled}. 
\item 	We extend our results to the scenario with continuous motion and arbitrary coordinates, provided the \nnew{distance between a robot's start and target positions} is at least one diameter; see Theorem~\ref{thm:continuousCaseUpperBound}.
This implies that efficient multi-robot coordination is always possible under relatively mild separability conditions; this includes non-convex robots.
\item For the continuous case of $N$ unit disks and weaker separability, we establish a lower bound of $\Omega(N^{1/4})$ and an upper bound 
of $\mathcal{O}(\sqrt{N})$ on the achievable stretch, see Theorem~\ref{thm:lowerBoundContinuousMakespan} and Theorem~\ref{thm:continuousCaseUpperBound}.

\end{itemize}

We also highlight the geometric difficulty of computing optimal trajectories even in seemingly simple cases; 
due to limited space, this can be found in Appendix~E.

\subsection{Related Work}	
Different variants of multiple-object motion planning problems have received a large amount of attention from researchers in various areas of computer science and engineering; see \cite{d-mprsm-07} for a survey.
Their practical relevance is reflected by the fact that there are industrial solutions used in automated warehouses for certain restricted forms of these problems \cite{wdm-chcavw-08}.
There are different orthogonal criteria by which these problems can be characterized.
A very important distinction is between \emph{discrete} and \emph{continuous} scenarios.
In the \emph{discrete} case, the input is a graph in which no two objects may use a vertex or edge at the same time;
depending on the scenario, it may be allowed to rotate fully populated cycles.
In the \emph{continuous} or \emph{geometric} setting, the objects are shapes in some geometric space which must be moved to a given target position in such a way that their interiors do not intersect at any time.
Depending on the scenario, the shapes may or may not touch.
Moreover, the objects may be confined to a certain region and there may be stationary obstacles.
Under these restrictions, it is unclear whether the target configuration is reachable at all.
Aronov~et~al.~\cite{abs+-mpmr-99} demonstrate that, for up to three robots of constant complexity, a path can be constructed efficiently if one exists.
Ramanathan and Alagar~\cite{ra-ampcmsdpo-85} as well as Schwartz and Sharir~\cite{ss-pmpcbpb-83} consider the case of several disk-shaped objects moving amongst polygonal obstacles.
They both find algorithms deciding whether a given target configuration is reachable.
Their algorithms run in time polynomial in the complexity of the obstacles, but exponential in the number of disks.
Hopcroft et al.~\cite{hss-cmpmio-84} and Hopcroft and Wilfong~\cite{hw-rmompgs-86} demonstrate the reachability of a given target configuration is PSPACE-complete to decide; this already holds when restricted to rectangular objects moving in a rectangular region.
Their proof was later generalized by Hearn and Demaine \cite{hd-pscsbpnclm-05,hd-gpcb-09}, who proved that rectangles of size $1 \times 2$ and $2 \times 1$ are sufficient and introduced a more general framework to prove PSPACE-hardness of certain block sliding games.
Moreover, this problem is similar to the well-known \emph{Rush Hour Problem}, which was shown to be PSPACE-complete by Flake and Baum~\cite{fb-rushhour-02}.
For moving disks, Spirakis and Yap~\cite{sy-snphmmd-84} have proven strong NP-hardness of the same problem; however, their proof makes use of disks of varying size.
Bereg et al.~\cite{bdp-sdp-08} as well as Abellanas et al.~\cite{abh+-mc-06} consider minimizing the \emph{number} of moves of a set of disks into a target arrangement without obstacles.
They provide simple algorithms and establish upper and lower bounds on the number of moves, where a move consists of sliding one disk along some curve without intersecting other disks.
These bounds were later improved on by Dumitrescu and Jiang \cite{dj-rdprp-13}, who also prove that the problem remains NP-hard for congruent disks even when the motion is restricted to sliding.

\revised{
Kirkpatrick and Liu~\cite{kl-cmlcm2d-16} consider the case of moving two disks of arbitrary radius from a start into a target configuration in an otherwise obstacle-free plane, minimizing the sum of distances travelled by the disks.
They provide optimal solutions for two disks moving from an arbitrary initial configuration into an arbitrary goal configuration.
Their arguments do not seem to generalize to the makespan.
}

D{\'{i}}az-B{\'{a}}{\~{n}}ez et al.~\cite{dhp+-cbpop-16} considered the task of extracting a single object from a group of convex objects, moving a minimal number of objects out of the way. They present an algorithm that finds the optimal direction for extracting the object in polynomial time.

On the practical side, there are several approaches to solving multi-object motion planning problems, both optimally and heuristically.
For discrete instances with a moderate number of objects, optimal solutions can be found using standard search strategies like $A^*$ \cite{hnr-astar-68} in the high-dimensional search space of possible configurations.
Numerous techniques can be used to improve the efficiency of these strategies \cite{s-oscpp-10,fgs+-peasng-12,gfs+-epea-14}.
Moreover, there is some work employing SAT solvers \cite{ks-usgp-99,hcz-ntbesps-10} to solve multi-object motion planning problems to optimality.
More recently, Yu and LaValle \cite{yl-omppgcaeh-16} present an IP-based exact algorithm for minimizing the makespan that works for hundreds of robots, even for challenging configurations with densities of up to 100\%.

For larger instances, one has to resort to heuristic solutions.
In \emph{priority planning} \cite{cl-sppcamm-80,el-mmo-87,bo-pmpmr-05,rl-cmrpphpa-06,gcb-bcp-06,dh-dppmtcc-12,cns+-adppcmrs-13}, the paths are planned one-by-one by assigning priorities to the objects and planning the movement in decreasing order of priority, treating all objects with higher priority as moving obstacles.
Kant and Zucker~\cite{kz-etppvd-86} decompose the problem into planning the paths for all objects and avoiding collisions by adapting the velocity of the objects appropriately, an approach which several papers are based on \cite{lh-ompig-98,lls-mpcmrga-99,sl-pcmmrrca-02,pa-cmrkcsp-05,cgg-pocmrsg-12}.
Another approach is to compute paths for the objects individually and resolve collisions locally \cite{fh-olsfpmrs-85}.

Between these simple \emph{decoupled} heuristics which only consider individual objects at a time and high-dimensional \emph{coupled} search algorithms lie \emph{dynamically-coupled} algorithms \cite{a-hlgmpmac-00,bsl+-cppmrodsp-09,ssf+-macbsomp-12,ssg+-ictsomp-13,bss+-svcbsmp-14,sh-kcmrmp-14,ssf+-cbsomp-15} which aim for better solutions at the price of higher computational costs.
These algorithms typically consider individual objects and only increase the dimension of the search space once a non-trivial interaction between objects is discovered.
Recently, Wagner and Choset~\cite{wc-sempp-15} provided a complete algorithm based on a similar principle.

With the advent of robot swarms, practical solutions to these problems became
more important and the robotics community started to develop practical
sampling-based algorithms
\cite{so-cppmr-98,sl-prmpccdpmrs-02,hh-hmpctdpop-04,wc-mstar-11,shh-pmsecsdnp-15,ssh-fnehdrrteirmrmp-15,sh-sbbpafm-16}
which, while working well in practice, are not guaranteed to find an (optimal)
solution.  In another recent work, Yu and Rus~\cite{yr-eafno-15} present a
practical algorithm based on a fine-grained discretization combined with an IP
for the resulting discrete problem to provide near-optimal solutions even for
densely populated environments. 
Other related work includes
Rubenstein et al.~\cite{rubenstein2014programmable}, who demonstrated how to
reconfigure a large swarm of simple, disk-shaped {\em Kilobots}; however, their
method is sequential, relocating one robot at a time, so a full reconfiguration of
1000 robots takes about a day, highlighting the 
relevance of truly parallel motion planning.
Further extensions to higher-dimensional problems (with a wide range of additional motion 
constraints) \textcolor{black}{are} swarms of drones (e.g., the work by Kumar~\cite{kumar}) and even air traffic control 
(see Delahaye et al.~\cite{delahaye2014mathematical} for
a recent survey).

In both discrete and geometric variants of the problem, the objects can be \emph{labeled}, \emph{colored} or \emph{unlabeled}.
In the \emph{labeled} case, the objects are all distinguishable and each object has its own, uniquely defined target position.
This is the most extensively studied scenario among the three.
In the \emph{colored} case, the objects are partitioned into $k$ groups and each target position can only be covered by an object with the right color.
This case was recently considered by Solovey and Halperin~\cite{sh-kcmrmp-14}, who present and evaluate a practical sampling-based algorithm.
In the \emph{unlabeled} case, the objects are indistinguishable and each target position can be covered by any object.
This scenario was first considered by Kloder and Hutchinson~\cite{kh-ppimf-06}, who presented a practical sampling-based algorithm.
In this situation, Turpin~et~al.~\cite{tmk-tpams-13} prove that it is possible to find a solution in polynomial time, if one exists.
This solution is optimal with respect to the longest distance traveled by any one robot.
However, their results only hold for disk-shaped robots under additional restrictive assumptions on the free space.
For unit disks and simple polygons, Adler et al.~\cite{adh+-emmpudsp-15} provide a \nnew{polynomial-time} algorithm under the additional assumption that the start and target positions have some minimal distance from each other.
Under similar separability assumptions, Solovey et al.~\cite{syz+-mpudog-15} provide a polynomial time-algorithm that produces a set of paths that is no longer than $\mbox{OPT}+4m$, where $m$ is the number of robots.
However, they do not consider the makespan, but only the total path length.
On the negative side, Solovey and Halperin~\cite{sh-hummp-15} prove that the unlabeled multiple-object motion planning problem is PSPACE-hard, even when restricted to unit square objects in a polygonal environment.

Regarding discrete multiple-object motion planning, C\v{a}linescu et al.~\cite{cdp-rgg-08} consider the non-parallel motion planning problem on graphs, where each object can be moved along an unoccupied path in one move.
They prove that both in the unlabeled and in the labeled case, minimizing the number of moves required is APX-hard.
They provide 3-approximation algorithms for the unlabeled case on general graphs.
Moreover, they prove that the problem remains NP-complete on the infinite rectangular grid.
Their results are different from our results because the objective they consider is not closely related to the makespan.
For other work, see \cite{dhm-mmfpt-14,bdz-oammp-11,dhmsoz-mm-09,hkkk-aamr-16} for particular examples.

On grid graphs, the problem can be cast as a very restrictive variant of mesh-connected routing, where each processor can only hold one packet at any time.
However, approaches developed for this problem (see Kunde~\cite{k-rsmca-88} and Cheung and Lau~\cite{cf-mprl-92}) typically assume that at least a constant number of packets can be held at any processor.
On the other hand, on grid graphs, the problem resembles the generalization of the 15-puzzle, for which Wagner \cite{w-gphag-74} and Kornhauser et al.~\cite{kms-cpmgdpga-84} have given an efficient algorithm that decides reachability of a target configuration and provided both lower and upper bounds on the number of moves required.
However, Ratner and Warmuth \cite{rw-fss15pi-86} proved finding a shortest solution for this puzzle remains NP-hard.
Demaine et al.~\cite{ddv-cmp-02} also consider various grids.
For the triangular grid, they give efficiently verifiable conditions for checking whether a solution exists.

\section{Preliminaries}\label{sec:pre}
In the grid setting \textcolor{black}{of} Section~\ref{sec:constappr} \textcolor{black}{we consider an $n_1 \times n_2$-grid $G=(V,E)$, which is dual to} 
an $n_1 \times n_2$-rectangle $P$ \textcolor{black}{in which the considered robots are arranged}.
A \emph{configuration} of $P$ is a mapping $C: V \rightarrow \{
1,\dots, N, \bot\}$, \new{which is injective w.r.t. the labels} $\{ 1,\dots,N \}$ of the $N \leq |P|$
robots to be moved, \new{where $\bot$ denotes the empty square}.
The inverse image of a robot's label~$\ell$ is~$C^{-1}(\ell)$.  In the following, we consider a \emph{start
configuration $C_s$} and \emph{target configuration $C_t$};  
for $i \in \{ 1,\dots,N \}$, we call $C_s^{-1}(i)$ and~$C_t^{-1}(i)$ the
\emph{start} and \emph{target position} of the robot $i$. \nnew{Given the (minimum) Manhattan distance between each robot's start/target positions for each robot, we denote by $d$ the maximum such distance over all robots.}

A configuration $C_1:V \rightarrow \{ 1,\dots,N,\bot \}$ can be
\emph{transformed \nnew{within one single transformation step}} into another configuration $C_2:V \rightarrow \{
1,\dots,N,\bot \}$, denoted $C_1 \rightarrow C_2$, if $C_1^{-1}(\ell) =
C_2^{-1}(\ell)$ or $(C_1^{-1}(\ell), C_2^{-1}(\ell)) \in E$ holds for all $\ell
\in \{ 1,\dots,N \}$, i.e., if each robot does not move or moves to one of the
\new{at most} four adjacent squares.  Furthermore, two robots cannot exchange their squares
in one transformation step, i.e., for all occupied squares $v \neq w \in V$, we
require that $C_2(v) = C_1(w)$ implies $C_2(w) \neq C_1(v)$. \textcolor{black}{For $M \in \mathbb{N}$, a \emph{schedule} is a sequence $C_1 \rightarrow \dots \rightarrow C_M$ of transformations.}
The number of steps in a schedule is called its
\emph{makespan}.  Given a start configuration $C_s$ and a target configuration
$C_t$, the \emph{optimal makespan} is the minimum number of steps in a
\textcolor{black}{schedule} starting with $C_s$ and ending with $C_t$.
\new{Let $n > 1$.} Note that for the $2 \times 2$-, $1 \times n$- and $n \times 1$-rectangles,
there are pairs of start and target configurations where no such sequence
exists.  For all other rectangles, such configurations do not exist; we
provide an $\mathcal{O}(1)$-approximation of the makespan in
Section~\ref{sec:constappr}.

For the continuous setting of Section~\ref{sec:cont}, we consider $N$ robots $R := \{ 1,\dots,N \} \subseteq \mathbb{N}$.
The Euclidean distance between two points $p,q \in \mathbb{R}^2$ is $|pq| := ||p-q||_2$.
Every robot $r$ has a \emph{start} and \emph{target position} $s_r,t_r \in \mathbb{R}^2$ with $|s_is_j|,|t_it_j| \geq 2$ for all $i \neq j$.
In the following, $d := \max_{r \in R} |s_rt_r|$ is the maximum distance a robot has to cover.
A \emph{trajectory} of a robot $r$ is a curve $m_r: [0,T_r] \rightarrow \mathbb{R}^2$, where $T_r \in \mathbb{R}^+$ denotes the \emph{travel time} of $r$.
This curve $m_r$ does not have to be totally differentiable, but must be totally left- and right-differentiable.
Intuitively, at any point in time, a robot has a unique \emph{past} and \emph{future direction} that are not necessarily identical.
This allows the robot to make sharp turns, but does not allow jumps.
We bound the speed of the robot by $1$, i.e., for each point in time, both left and right derivative of $m_r$ have Euclidean length at most $1$.
Let $m_i:[0,T_i] \rightarrow \mathbb{R}^2$ and $m_j:[0,T_j] \rightarrow \mathbb{R}^2$ be two trajectories;
w.l.o.g., all travel times are equal to the maximum travel time \nnew{$T_{\max}$} by \nnew{extending $m_r$ with} $m_r(t) = m_r(T_r)$ for all $T_r < t \leq T_{\max}$.
\textcolor{black}{The trajectories} $m_i$ and $m_j$ are \emph{compatible} if the corresponding robots do not intersect at any time, i.e., if $|m_i(t)m_j(t)| \geq 2$ holds for all $t \in [0,T_i]$.
A \emph{trajectory set} of $R$ is a set of compatible trajectories $\{ m_1,\dots,m_N \}$, one for each robot.
The \emph{(continuous) makespan} of a trajectory set $\{ m_1,\dots,m_N \}$ is \textcolor{black}{defined} as $\max_{r \in R} T_r$.
A trajectory set $\{ m_1,\ldots,m_N \}$ \emph{realizes} a pair of start and target configurations $\mathcal{S} := (\{ s_1,\ldots,s_N \},\{ t_1,\ldots,t_N \})$ if $m_r(0) = s_r$ and $m_r(T_r) = t_r$ hold for all $r \in R$.
We are searching for a trajectory set $\{ m_1,\dots,m_N \}$ realizing $\mathcal{S}$ with minimal makespan.

\section{Labeled Grid Permutation}\label{sec:constappr}
\textcolor{black}{Let $n_1 \geq n_2 \geq 2$, $n_1 \geq 3$} and let $P$ be an $n_1 \times n_2$-rectangle.
In this section, we show that computing the optimal makespan of arbitrarily chosen start and target configurations $C_s$ and $C_t$ of 
$k$ robots in $P$ is strongly NP-complete. \new{This is followed by a $\mathcal{O}(1)$-approximation for the makespan.}

\begin{theorem}\label{thm:optimalNPcomplete}
	The minimum makespan parallel motion planning problem on a grid is strongly NP-hard.
\end{theorem}
\nnew{
We prove hardness using a reduction from \textsc{Monotone 3-Sat}.
Intuitively speaking, given a formula, we construct a parallel motion planning instance with a \emph{variable robot} for each variable in the formula.
To encode a truth assignment, each variable robot is forced to move on one of two paths.
This is done by employing two groups of auxiliary robots that have to move towards their goal in a straight line in order to realize the given makespan.
These auxiliary robots form \emph{moving obstacles} whose position is known at any point in time.

The variable robots cross paths with \emph{checker robots}, one for each literal of the formula, forcing the checker to wait for one time step if the assignment does not satisfy the literal.
The checker robots then cross paths with \emph{clause robots}; each clause robot has to move to its goal without delay and can only do so if at least one of the checkers did not wait.
In order to ensure that the checkers meet with the clauses at the right time, further auxiliary robots force the checkers to perform a sequence of side steps in the beginning.
Figure~\ref{fig:reduction-overview-main} gives a rough overview of the construction; full details of the proof are given in Section~\ref{sec:labeledgridperm_details}.
}
\begin{figure}
	\begin{center}
		\resizebox{.5\linewidth}{!}{\includegraphics{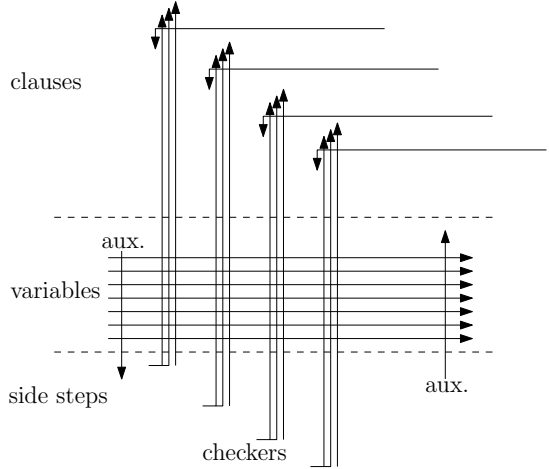}}
	\end{center}
	\caption{A sketch of the parallel motion planning instance resulting from the reduction.}
	\label{fig:reduction-overview-main}
\end{figure}

\nnew{In the proof of NP-completeness, we use a pair of start and target configurations in which the corresponding grids are not fully occupied.
However, for our constant-factor approximation, we assume in Theorem~\ref{thm:main} that the grid is fully occupied.
This assumption is without loss of generality; our approximation algorithm works for any grid population, see Theorem~\ref{thm:main2}.}

\new{Our constant-factor approximation is based on an algorithm that computes a schedule with a makespan upper-bounded by $\mathcal{O}(n_1+n_2)$ described by Lemma~\ref{lem:naive}. Based on Lemma~\ref{lem:naive}, we give a constant factor approximation of the makespan, see Theorem~\ref{thm:main}. Finally, we embed the algorithm of Theorem~\ref{thm:main} into a more general approach to ensure simultaneously a polynomial running time w.r.t.~the number $N$ of input robots and a constant approximation factor, see Theorem~\ref{thm:main2}.}

\begin{lemma}\label{lem:naive}
	For a pair of start and target configurations $C_s$ and $C_t$ of an $n_1 \times n_2$-rectangle, we can compute in polynomial time \new{w.r.t.~$n_1$ and $n_2$} a sequence of $\mathcal{O}(n_1+n_2)$ steps transforming~$C_s$ into~$C_t$.
\end{lemma}

\new{The high-level idea of the algorithm of Lemma~\ref{lem:naive} is the
following. We apply a sorting algorithm called {\sc
RotateSort}~\cite{marberg:sorting} that computes a corresponding permutation of
an $n_1 \times n_2$ (orthogonal) grid within $\mathcal{O}(n_1+n_2)$
\emph{parallel steps}. Each parallel step is made up of a set of pairwise
disjoint \emph{swaps}, each of which causes two neighbouring robots to exchange
their positions. Because in our model direct swaps are not allowed, we simulate one
parallel step by a sequence of $\mathcal{O}(1)$ transformation steps. This
still results in a sequence of $\mathcal{O}(n_1+n_2)$ transformation steps.
A detailed description of the algorithm used in the proof of
Lemma~\ref{lem:naive} is given in Section~\ref{sec:scheduleNaive}}.

\new{Based on the algorithm of Lemma~\ref{lem:naive}, we can give a constant-factor approximation algorithm.}

\begin{theorem}\label{thm:main}
	There is an algorithm with running time $\mathcal{O}(dn_1n_2)$ that, given an arbitrary pair of start and target configurations of an $n_1 \times n_2$-rectangle with
maximum distance $d$ between any start and target position, computes a schedule of makespan $\mathcal{O}(d)$, i.e., an approximation algorithm with constant stretch.
\end{theorem}

For the algorithm of Theorem~\ref{thm:main}, Lemma~\ref{lem:naive} is repeatedly applied to rectangles of side length~$\mathcal{O}(d)$, resulting in $\mathcal{O}(d)$ transformation steps in total.
Because $d$ is a lower bound on the makespan, this yields an $\mathcal{O}(1)$-approximation of the makespan.

\nnew{At} a high level, the algorithm of Theorem~\ref{thm:main} first computes the maximal Manhattan distance~$d$ between a robot's start and target position.
Then we partition $P$ into a set $T$ of pairwise disjoint rectangular \emph{tiles}, where each tile $t \in T$ is an $n_1' \times n_2'$-rectangle for $n_1',n_2' \leq 24d$.
We then use an algorithm based on flows to \new{compute a sequence of $\mathcal{O}(d)$ transformation steps, ensuring} that all robots are in their target tile.
Once all robots are in the correct tile, we use Lemma~\ref{lem:naive} simultaneously on all tiles to move each robot to the correct position within its target tile. \new{The details of the algorithm of Theorem~\ref{thm:main} are given further down in this section.}

The \new{above mentioned tiling} construction \new{ensures} that each square of $P$ belongs to one unambiguously defined tile and each robot has a \emph{start} and \emph{target tile}.

\new{
Based on the approach of Theorem~\ref{thm:main} we give a $\mathcal{O}(1)$-approximation algorithm for the makespan with a running time polynomial w.r.t.~the number $N$ of robots to be moved.

\begin{theorem}\label{thm:main2}
	There is an algorithm with running time $\mathcal{O}(N^5)$ that, given an arbitrary pair of start and target configurations of a rectangle $P$ with $N$ robots to be moved and maximum distance $d$ between any start and target position, computes a schedule of makespan $\mathcal{O}(d)$, i.e., an approximation algorithm with constant stretch.
\end{theorem}

	Intuitively speaking, the approach of Theorem~\ref{thm:main2} distinguishes two cases. 

(1)~Both~$\lfloor \frac{n_1}{4} \rfloor$ and the maximum distance $d$ between the robots' start and target positions, are lower-bounded by the number $N$ of input robots. 

(2) $N > \lfloor \frac{n_1}{4} \rfloor$ or $N > d$.
	
	In case (1), the grid is populated sparsely enough such that the robots' trajectories in northern, eastern, southern, and western direction can be done sequentially by four individual transformation sequences, see Figure~\ref{fig:achievingIntermediate}.
	
	
	 In order to ensure that each robot has locally enough space, we consider a preprocessed start configuration $C_o$ in which the robots have odd coordinates. We ensure that $C_s$ can be transformed into $C_o$ within $\mathcal{O}(d)$ steps. Analogously, we ensure that the outcome of the northern, eastern, southern, and western trajectories is a configuration $C_e$ with even coordinates, such that $C_e$ can be transformed into $C_t$ within $\mathcal{O}(d)$ transformation steps.
	
	In the second case, we apply the approach of Theorem~\ref{thm:main} as a subroutine to a union of smallest rectangles that contain the robots' start and target configurations, see Figure~\ref{fig:polyTimeSecondApproach}.
	
	The full detailed version of the proof of Theorem~\ref{thm:main2} can be found in Section~\ref{sec:polymain2}.

}

\new{In the rest of Section~\ref{sec:constappr}, we give the proof of Theorem~\ref{thm:main}, i.e. \nnew{we give} an algorithm that computes a schedule with makespan linear in the maximum distance between robots' start and target positions.} The remainder of the proof \new{of Theorem~\ref{thm:main}} is structured as follows.
In Section~\ref{sec:outline} we give an outline of our \emph{flow algorithm} that ensures that each robot reaches its target tile in $\mathcal{O}(d)$ transformation steps.
Section~\ref{sec:details} \new{gives the full intuition} of this algorithm and
its subroutines. \new{(For full details, we refer to
Section~\ref{sec:labeledgridperm_details})}.

\subsection{Outline of the Approximation Algorithm \new{of Theorem~\ref{thm:main}}}\label{sec:outline}
We model the trajectories of robots between tiles as a flow $f_T$, using the weighted directed graph $G_T = (T, E_T, f_T)$, which is dual to the tiling $T$ defined in the previous section.
In $G_T$, we have an edge $(v,w) \in E_T$ if there is at least one robot that has to move from $v$ into $w$.
Furthermore, we define the weight $f_T((v,w))$ of an edge as the \textcolor{black}{integer} number of robots that move from $v$ to $w$.
As $P$ is fully occupied, $f_T$ is a \emph{circulation}, i.e., a flow with no sources or sinks, in which flow conservation has to hold at all vertices.
\nnew{Because the side lengths of the tiles are greater than~$d$,} $G_T$ is a grid graph with additional diagonal edges and thus has degree at most $8$.

The maximum edge value of $f_T$ is $\Theta (d^2)$, but only $\mathcal{O}(d)$ robots can possibly leave a tile within a single transformation step.
Therefore, we decompose the flow $f_T$ of robots into a \emph{partition} consisting of $\mathcal{O}(d)$ \emph{subflows}, where each individual robot's motion is modeled by exactly one subflow and each edge in the subflow has value at most $d$.
\textcolor{black}{Thus we are able to \emph{realize} each subflow in a single transformation step by placing the corresponding robots adjacent to the boundaries of its corresponding tiles before we realize the subflow}.
To facilitate the decomposition into subflows, we first preprocess~$G_T$.
In total, the algorithm consists of the following subroutines, elaborated in detail in Section~\ref{sec:details}.
\begin{itemize}
	\item \bf Step 1: Compute $d$, the tiling $T$ and the corresponding flow $G_T$.
	\item \bf Step 2: Preprocess $G_T$ in order to remove intersecting and bidirectional edges.
	\item \bf Step 3: Compute a partition into $\mathcal{O}(d)$ $d$-subflows.
	\item \bf Step 4: Realize the $\mathcal{O}(d)$ subflows using $\mathcal{O}(d)$ transformation steps.
	\item \bf Step 5: Simultaneously apply Lemma~\ref{lem:naive} to all tiles, moving each robot to its target position.
\end{itemize}

\subsection{Details of the Approximation Algorithm \new{of Theorem~\ref{thm:main}}}\label{sec:details}
\new{In this section we only give more detailed descriptions of Steps~1-4 because Step 5 is a trivial application of Lemma~\ref{lem:naive} to all tiles in parallel.}

\subsubsection{Step 1: Compute $d$, the Tiling $T$, and the corresponding Flow $G_T$}\label{sec:tiling}

\new{The maximal distance between robots' start and target positions can be computed in a straightforward manner.}

\new{For the tiling, we assume that the rectangle $P$ is axis aligned and that its \nnew{bottom-left} corner is~$(0,0)$. We consider $k_v := \lfloor \frac{n_1}{12d} \rfloor$ vertical lines $\ell^v_1,\dots,\ell^v_{k_v}$ with $x$-coordinate modulo $12d$ equal to $0$. Analogously, we consider $k_h:=\lfloor \frac{n_2}{12d} \rfloor$ horizontal lines $\ell^h_1,\dots,\ell^h_{k_h}$ with $y$-coordinate modulo $12d$ to $0$. Finally, we consider the tiling of $P$ that is induced by the arrangement induced by $\ell^v_1,\dots,\ell^v_{k_v -1}, \ell^h_1,\dots,\ell^h_{h_v -1}$ and the boundary of $P$, see Figure~\ref{fig:tiling}.} \nnew{This implies that the side length of a tile is upper-bounded by $24d - 1$.}

\new{Finally, computing the flow $G_T$ is straightforward by considering the tiling $T$ and the robots' start and target positions.}

 
\subsubsection{Step 2: Ensuring Planarity and Unidirectionality}
After initialization, we preprocess $G_T$, removing edge intersections and bidirectional edges by transforming the start configuration $C_s$ into an intermediate start configuration $C_s'$, obtaining a planar flow without bidirectional edges.
This transformation consists of two steps: (1) ensuring planarity and (2) ensuring unidirectionality.

\begin{figure}[ht]

\begin{center}
\begin{tabular}{cccc}
\includegraphics[scale=0.4]{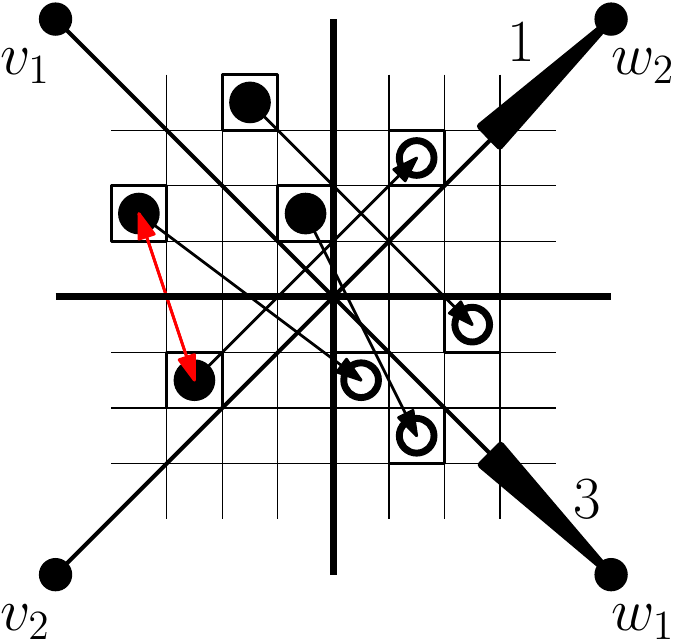}  &
\includegraphics[scale=0.4]{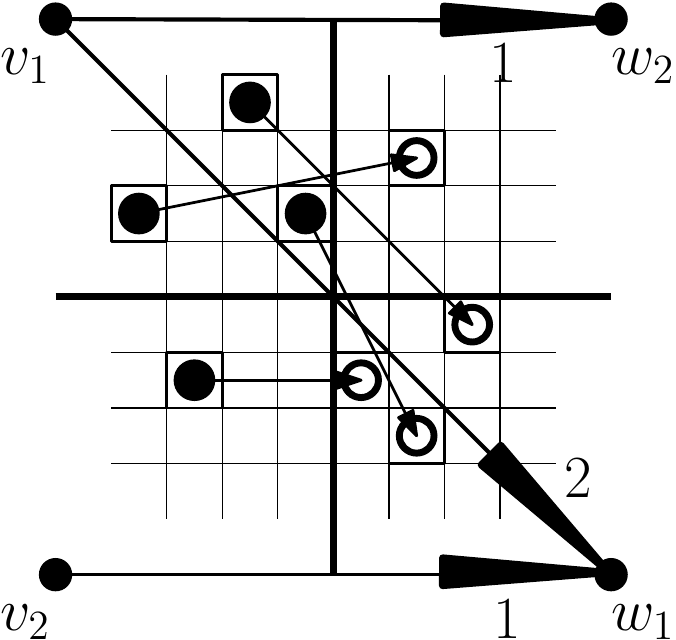} &
\includegraphics[scale=0.4]{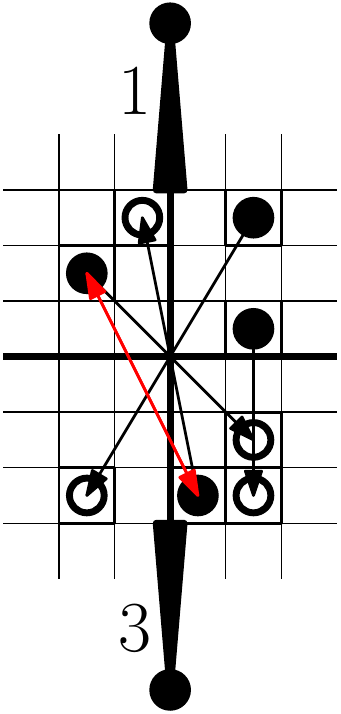}  &
\includegraphics[scale=0.4]{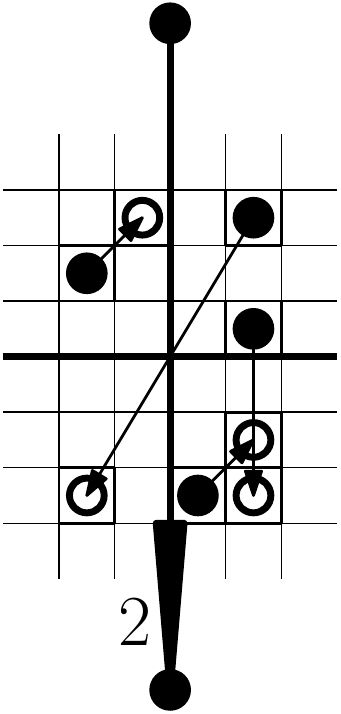} \\
{\small (a)}  &
{\small (b) }&
{\small (c) }&
{\small (d) }\\
\end{tabular}
\end{center}
\vspace*{-12pt}
\caption{Illustration of the preprocessing (step (1): before and after removing crossing edges (a)+(b) and step (2): before and after removing bidirectional edges (c)+(d)). The red arrows indicate how robots change their positions during the preprocessing steps.}
\label{fig:crossing}
\vspace*{-10pt}
\end{figure}

\textbf{Step (1):} We observe that edge crossings only occur between two diagonal edges with adjacent source tiles, as illustrated in Figure~\ref{fig:crossing}(a)+(b).
To remove a crossing, it suffices to eliminate one of the diagonal edges by exchange robots between the source tiles.
To eliminate all crossings, each robot is moved at most once, because after moving, the robot does no longer participate in a diagonal edge.
Thus, all necessary exchanges can be done in $\mathcal{O}(d)$ steps by Lemma~\ref{lem:naive}, covering the tiling $T$ by constantly many layers, similar to the proof of Lemma~\ref{lem:naive}.

\begin{sloppypar}
\textbf{Step (2):} We delete a bidirectional edge $(v,w),(w,v)$ by moving $\min \{ f_T((v,w)),\,f_T((w,v)) \}$ robots with target tile $w$ from $v$ to $w$ and vice versa \nnew{which achieves that $\min \{ f_T((v,w)),\,f_T((w,v)) \}$ robots achieve their target tile $w$ and $\min \{ f_T((v,w)),\,f_T((w,v)) \}$ robots achieve their target tile $v$,} thus eliminating the edge with lower flow value.
This process is depicted in Figure~\ref{fig:crossing}(c)+(d).
Like step (1), this can be done in $\mathcal{O}(d)$ parallel steps by Lemma~\ref{lem:naive}.
As we do not add any edges, we maintain planarity during step (2).
Observe that during the preprocessing, we do not destroy the grid structure of $G_T$.
\end{sloppypar}

\new{Step (1) and step (2) maintain the flow property of $f_T$ without any other manipulations to the flow $f_T$, because both preprocessing steps can be represented by local circulations.}

 

\subsubsection{Step 3: Computing a Flow Partition}\label{sec:flowPartition}
After preprocessing, we partition the flow $G_T$ into $d$-subflows.

\begin{definition}\label{def:subflow}
	A \emph{subflow} of $G_T$ is a \textcolor{black}{circulation} $G'_T = (T, E',f'_T)$, such that $E' \subseteq E_T$, and $0 \leq f_T'(e) \leq f_T(e)$ for all $e \in E'$. 
	If $f_T'(e) \leq z$ for all $e \in E'$ and some $z \in \mathbb{N}$, we call $G'_T$ a $z$-flow.
\end{definition}

The flow partition relies on an upper bound on the maximal edge weight in $G_T$.
By construction, tiles have side length at most $24d$; therefore, each tile consists of at most $576d^2$ unit squares.
This yields the following upper bound; a tighter constant factor can be achieved using a more sophisticated argument.

\begin{observation}\label{obs:maxEdgeWeight}
	We have $f_T(e) \leq 576d^2$ for all $e \in E_T$.
\end{observation}

\begin{definition}\label{def:partition}
	A $(z,\ell)$-partition of $G_T$ is a set of $\ell$ $z$-subflows $\{ G_1=(V_1,E_1,f_1),\dots,G_{\ell}=(V_{\ell},E_{\ell},f_{\ell}) \}$ of $G_T$, such that $G_1,\dots,G_{\ell}$ sum up to $G_T$.
\end{definition}

\begin{lemma}
\label{lem:computePartition}
We can compute a $(d,\mathcal{O}(d))$-partition of $G_T$ in polynomial time.
\end{lemma}
\begin{proof}[Proof sketch]
In a slight abuse of notation, throughout this proof, the elements in sets of cycles are not necessarily unique.
A $(d,\mathcal{O}(d))$-partition can be constructed using the following steps.
\begin{itemize}
        \item We start by computing a $(1,h)$-partition $\mathbb{C}_{\bigcirc}$ of $G_T$ consisting of $h \leq n_1n_2$ cycles.
        This is possible because $G_T$ is a circulation.
        If a cycle $C$ intersects itself, we subdivide $C$ into smaller cycles that are intersection-free.
        Furthermore, $h$ is clearly upper bounded by the number of robots $n_1n_2$, because every robot can contribute only $1$ to the sum of all edges in $G_T$.
        As the cycles do not self-intersect, we can partition the cycles $\mathbb{C}_{\bigcirc}$ by their orientation, obtaining the set $\mathbb{C}_{\circlearrowright}$ of clockwise and the set $\mathbb{C}_{\circlearrowleft}$ of counterclockwise cycles.

        \item We use $\mathbb{C}_{\circlearrowright}$ and $\mathbb{C}_{\circlearrowleft}$ to compute a $(1,h')$-partition $\mathbb{C}_{\circlearrowright}^1 \cup \mathbb{C}_{\circlearrowright}^2 \cup \mathbb{C}_{\circlearrowleft}^1 \cup \mathbb{C}_{\circlearrowleft}^2$ with $h' \leq n_1n_2$, such that two cycles from the same subset $\mathbb{C}_{\circlearrowright}^1$,  $\mathbb{C}_{\circlearrowright}^2$, $\mathbb{C}_{\circlearrowleft}^1$, or $\mathbb{C}_{\circlearrowleft}^2$ share a common orientation.
        Furthermore, we guarantee that two cycles from the same subset are either edge-disjoint or one lies nested in the other.
        A partition such as this can be constructed by applying a recursive peeling algorithm to $\mathbb{C}_{\circlearrowright}$ and $\mathbb{C}_{\circlearrowleft}$ as depicted in Figure~\ref{fig:intersectionFree}, yielding a decomposition of the flow induced by $\mathbb{C}_{\circlearrowright}$ into two cycle sets $\mathbb{C}_{\circlearrowright}^1$ and $\mathbb{C}_{\circlearrowright}^2$, where $\mathbb{C}_{\circlearrowright}^1$ consists of clockwise cycles and $\mathbb{C}_{\circlearrowright}^2$ consists of counterclockwise cycles, and a similar partition of $\mathbb{C}_{\circlearrowleft}$\nnew{, see the appendix for details.} 

        \begin{figure}[h]
                \begin{center}
                        \resizebox{.75\linewidth}{!}{\includegraphics{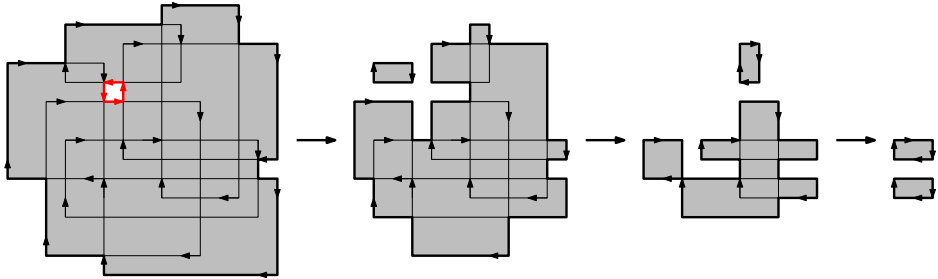}}
                \end{center}
                \caption{Recursive peeling of the area bounded by the cycles from $\mathbb{C}_{\circlearrowright}$, resulting in clockwise cycles (thick black cycles).
                        Cycles constituting the boundary of \emph{holes} are counterclockwise (thick red cycles).
                        Note that an edge $e$ vanishes when $f_T(e)$ cycles containing that edge are removed by the peeling algorithm described above.}
        \label{fig:intersectionFree}
        \end{figure}

        \item Afterwards, we partition each set $\mathbb{C}_{\circlearrowright}^1$, $\mathbb{C}_{\circlearrowright}^2$, $\mathbb{C}_{\circlearrowleft}^1$, and $\mathbb{C}_{\circlearrowleft}^2$ into $\mathcal{O}(d)$ subsets, each inducing a $d$-subflow of $G_T$\nnew{, see the appendix for details.}
        \end{itemize}

 	\noindent 
 \end{proof}

\subsubsection{\textcolor{black}{A Subroutine of} Step 4: Realizing a Single Subflow}\label{sec:realizseOneFlow}
In this section, we present a procedure for \emph{realizing} a single $d$-subflow $G_T'$ of $G_T$.

\begin{definition}
	A \textcolor{black}{schedule $t := C_1 \rightarrow \dots \rightarrow C_{k+1}$} \emph{realizes} a subflow $G'_T = (T,E',f'_T)$ if, for each pair $v,w$ of tiles, the number of robots moved by $t$ from their start tile $v$ to their target tile $w$ is $f_T'((v,w))$, where we let $f_T'((v,w)) = 0$ if $(v,w) \notin E'$.
\end{definition}

\begin{lemma}
\label{lem:transStep}
Let $G'_T=(T,E_T',f_T')$ be a planar unidirectional $d$-subflow.
	There is a polynomial-time algorithm that computes a \textcolor{black}{schedule $C_1 \rightarrow \dots \rightarrow C_{k+1}$ realizing $G'_T$ for a constant $k \in \mathcal{O}(1)$}.
\end{lemma}
\begin{proof}[Proof sketch]
\new{We give a high-level description of the proof and refer to Section~\ref{sec:RealizSingleSubflow} for details.}

\begin{figure}[ht]
\begin{center}
	\begin{subfigure}[b]{.25\linewidth}
		\resizebox{\linewidth}{!}{\includegraphics{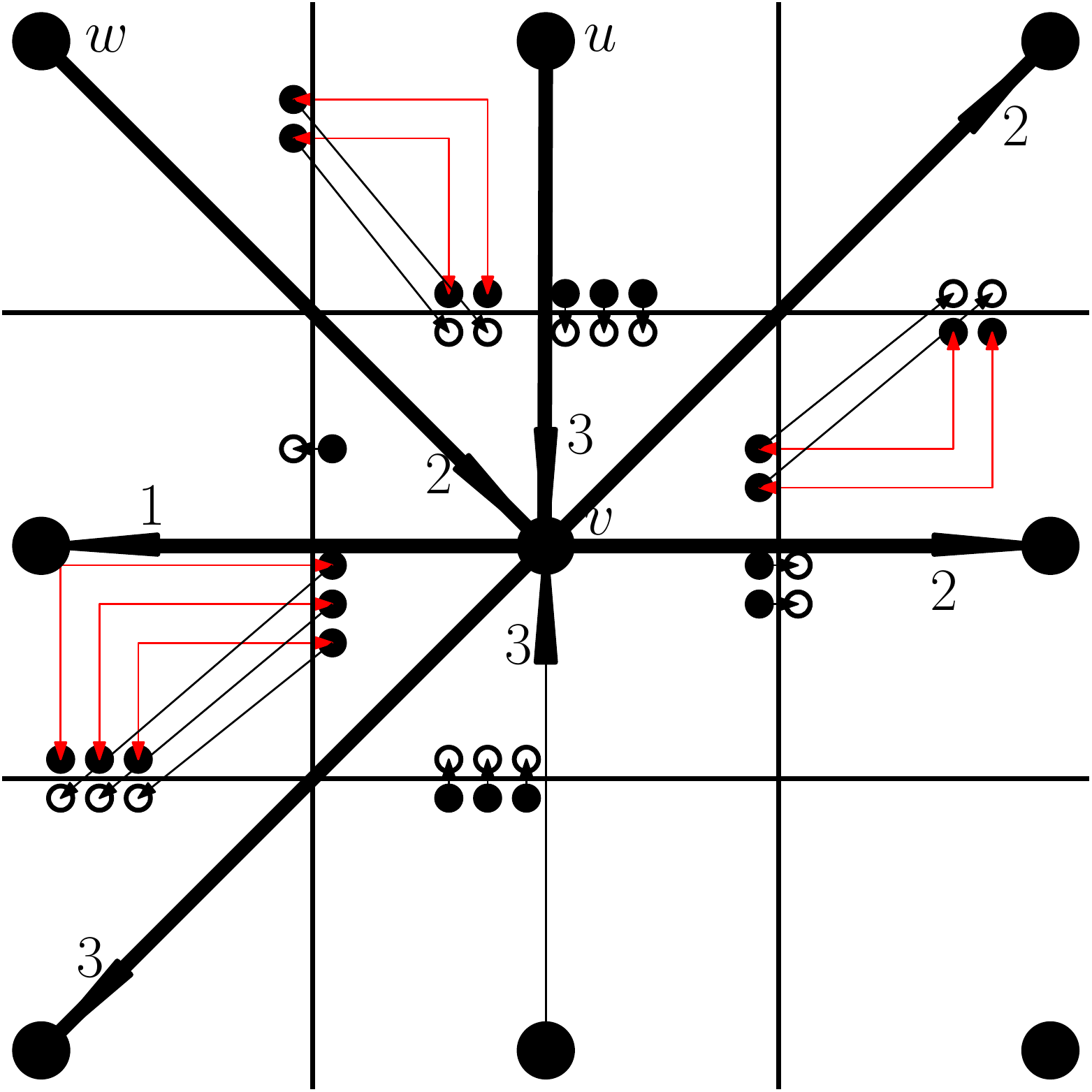}}
		\caption{Preprocessing of diagonal edges.}
	\end{subfigure}
	\hspace{.03\linewidth}
	\begin{subfigure}[b]{.25\linewidth}
		\resizebox{\linewidth}{!}{\includegraphics{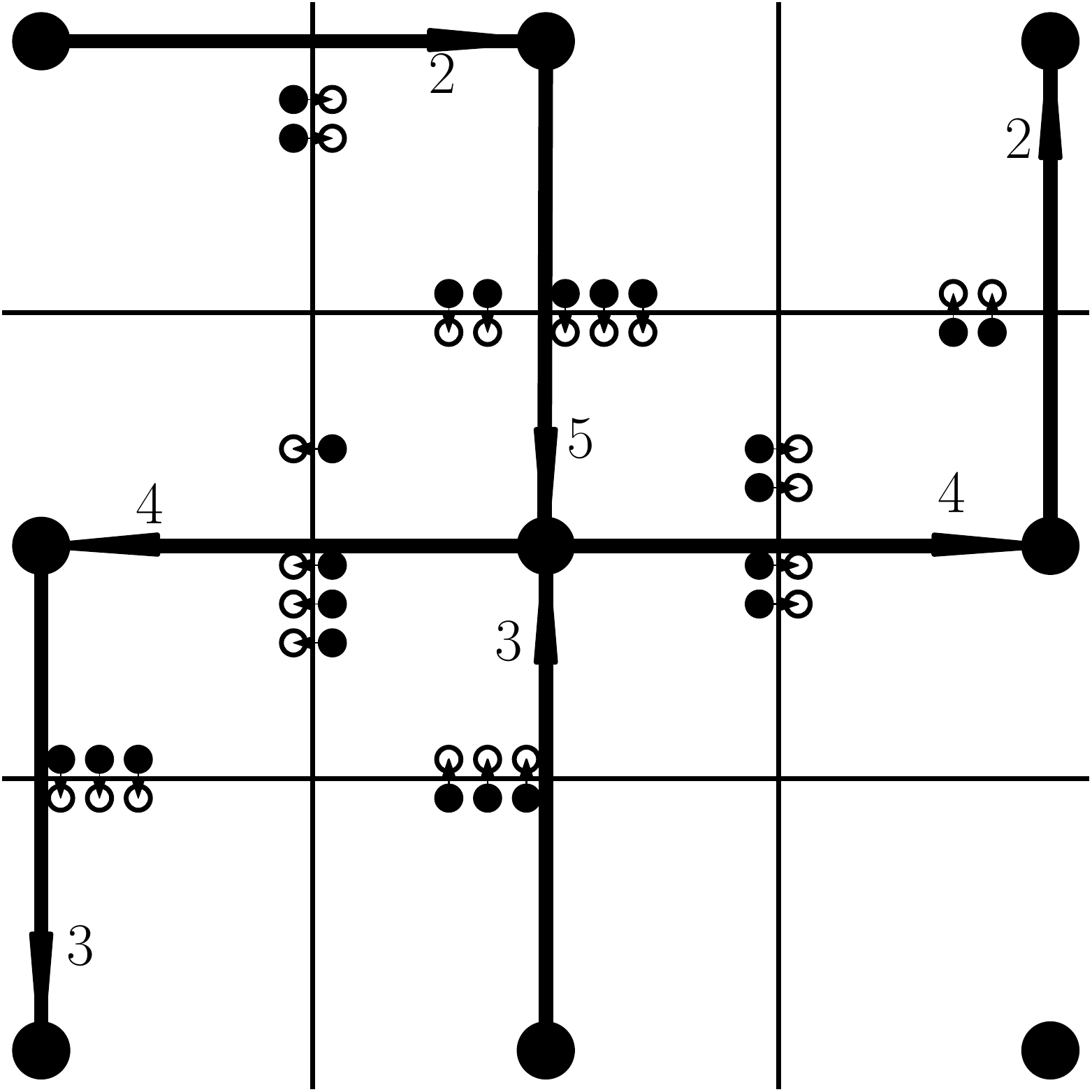}}
		\caption{Configuration and flow after preprocessing.}
	\end{subfigure}
	\hspace{.03\linewidth}
	\begin{subfigure}[b]{.35\linewidth}
		\resizebox{\linewidth}{!}{\includegraphics{figures/transStep.pdf}}
		\caption{A crossing-free matching of incoming and outgoing robots and the connecting paths inside the corresponding tile, for $d=3$.}
	\end{subfigure}
\end{center}
\vspace*{-12pt}
\caption{Procedure for computing transformation steps that realize a $d$-subflow. Figures (a) and~(b) illustrate how we preprocess $G_T'$ such that $E_T'$ consists of horizontal and vertical edges only. Figure (c) illustrates the main approach. \nnew{White disks illustrate start positions and black disks illustrate target positions.}}
\label{fig:transStepA}
\end{figure}

Our algorithm uses $k = \mathcal{O}(d)$ preprocessing steps $C_1 \rightarrow \dots \rightarrow C_k$, as depicted in Figure~\ref{fig:transStepA}(a)+(b), and one final realization step $C_k \rightarrow C_{k+1}$, shown in Figure~\ref{fig:transStepA}(c), pushing the robots from their start tiles into their target tiles.
The preprocessing eliminates diagonal edges and places the moving robots next to the border of their target tiles. \new{For the final realization step we compute a pairwise disjoint matching between incoming and outgoing robots, such that each pair is connected by a tunnel inside the corresponding tile in which these tunnels do not intersect, see Figure~\ref{fig:transStepA}(a). The final realization step is given via the robots' motion induced by pushing each robot into the interior of the tile and by pushing this one-step motion through the corresponding tunnel into the direction of the corresponding outgoing robot.}
\end{proof}

\subsubsection{Step 4: Realizing All Subflows}\label{sec:iterativelyRealise}
Next we extend the idea of Lemma~\ref{lem:transStep} to $\ell \leq d$ subflows instead of one and demonstrate how this can be leveraged to move all robots to their target tile using $\mathcal{O}(d)$ transformation steps.

\begin{lemma}
\label{lem:sequenceTransStep}
Let $\mathcal{S} := \langle G_1=(V_1,E_1,f_1),\dots,G_{\ell}=(V_{\ell},E_{\ell},f_{\ell}) \rangle$ be a sequence of $\ell \leq d$ unidirectional planar $d$-subflows of $G_T$.
	There is a polynomial-time algorithm computing $\mathcal{O}(d)  + \ell$ transformation steps $C_1 \rightarrow \dots \rightarrow C_{k+\ell}$ realizing $\mathcal{S}$.
\end{lemma}
\begin{proof}[Proof sketch]
	\new{We give a high level description of the proof and refer for details to Section~\ref{sec:realizAllSubflows}.}

        Let $t$ be an arbitrary tile.
        Similar to the approach of Lemma~\ref{lem:transStep}, we first apply a preprocessing \textcolor{black}{step} guaranteeing that the robots to be moved into or out of $t$ are in the right position close to the boundary of $t$\nnew{, see Figure~\ref{fig:sequenceSubflowsASketch}}.
        Thereafter we move the robots into their target tiles, using $\ell$ applications of the algorithm from Lemma~\ref{lem:transStep} without the preprocessing phase. \new{In particular, we realize a sequence of $\ell$ $d$-subflows by applying $\ell$ times the single realization step of Lemma~\ref{lem:transStep}.}
        
        
        \begin{figure}[h]
		\begin{center}
			\includegraphics[scale=0.3]{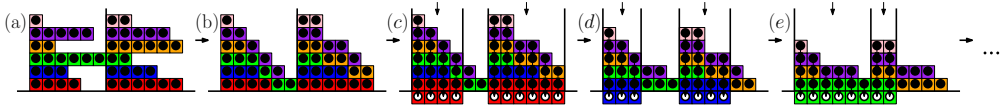}
		\end{center}
		\vspace{-3ex}
		\caption{Stacking robots in lines induced by flows of the edges of the subflows to be realized.}
		\label{fig:sequenceSubflowsASketch}
		\vspace{-3ex}
	\end{figure}

\end{proof}

\begin{lemma}\label{lem:iterativelyRealise}
	There is a polynomial-time algorithm computing $\mathcal{O}(d)$ transformation steps moving all robots into their target tiles.
\end{lemma}
\begin{proof}
	By Lemma~\ref{lem:computePartition}, we can compute a $(d, c d)$-partition of $G_T$ for $c \in \mathcal{O}(1)$.
	We group the corresponding $d$-subflows into $\frac{cd}{d}=c$ sequences, each consisting of at most $d$ $d$-subflows.
	We realize each sequence by applying Lemma~\ref{lem:sequenceTransStep}, using $\mathcal{O}(d)$ transformation steps for each sequence.
	This leads to $\mathcal{O}(cd) = \mathcal{O}(d)$ steps for realizing all sequences of $d$-subflows.
\end{proof}

For the proof of Theorem~\ref{thm:main}, we still need to analyze the time complexity of our approach, for which we refer to Section~\ref{sec:togther}.

\section{Variants on Labeling}\label{sec:labeling}
A different version is the unlabeled variant, in which all robots are the same.
A generalization of both this and the labeled version arises when robots belong to one of $k$ color classes,
with robots from the same color class being identical.

\new{We formalize this problem variant by
using a coloring $c: \{ 1,\dots,n_1n_2 \} \rightarrow \{
1,\ldots,k \}$ for grouping the robots.
By populating unoccupied cells with robots carrying color $k+1$, we may assume that
each unit square in the environment $P$ is occupied.
The robots \emph{draw an image} $I = \big(I^1,\ldots,I^k\big)$, where $I^i$ is the set of cells occupied by a robot with color $i$. We say that two images $I_s$ and $I_t$ are \emph{compatible} if in $I_s$ and $I_t$ the number of cells colored with color $i$ are equal for each color $i = 1,\dots,k$.
By moving the robots, we want to transform a start image $I_s$ into a \nnew{compatible} target image $I_t$, minimizing the makespan.}

\begin{theorem}\label{thm:unlabeled}
There is an algorithm with running time $\mathcal{O}(k(N)^{1.5} \log (N) + N^5)$ for computing, given start and target images $I_s,I_t$
with maximum distance~$d$ between start and target positions, 
an $\mathcal{O}(1)$-approximation of the optimal makespan $M$ and a corresponding \textcolor{black}{schedule}.
\end{theorem}

The basic idea is to transform the given unlabeled problem setting into a labeled problem setting by solving a geometric bottleneck matching problem\nnew{, see the appendix for details}.

\section{Continuous Motion}
\label{sec:cont}
The continuous case considers $N$ unit disks 
that have to move into a target configuration;
the velocity of each robot is bounded by 1, and we want to minimize the makespan. 
For arrangements of disks that are not well separated, we show that constant stretch \nnew{is {\em impossible}}.

\begin{theorem}\label{thm:lowerBoundContinuousMakespan}
	There is an instance with optimal makespan $M \in \Omega (N^{1/4})$, see Figure~\ref{fig:lowerBoundStartAndTarget}.
\end{theorem}


The basic proof idea 
is as follows. Let $\{ m_1,\dots,m_N \}$ be an arbitrary trajectory set with
makespan $M$.  We show that there must be a point in time $t \in [0,M]$ where
the area of $\CHull(m_1(t),\dots,m_N(t))$ is lower-bounded by
$cN+\Omega(N^{3/4})$, where $cN$ is the area of the convex hull
$\CHull(m_1(0),\dots,m_N(0))$ of $m_1(0),\dots,m_N(0)$.
Assume $M \in o \left( N^{1/4} \right)$ and consider the area of $\CHull(m_1(t'),\dots,m_N(t'))$ at some point $t' \in [0,M]$.
This area is at most $cN + \mathcal{O}(\sqrt{N}) \cdot o\left(N^{1/4} \right)$ which is a contradiction. Proof details are given in Section~\ref{sec:lowerBoundCont}.

Conversely, we give a non-trivial upper bound on the stretch, as follows.

\begin{theorem}\label{thm:continuousCaseUpperBound}
	There is an algorithm that computes a trajectory set with continuous makespan \nnew{of} $\mathcal{O}(d+\sqrt{N})$. If $d \in \Omega(1)$, this implies a $\mathcal{O}(\sqrt{N})$-approximation algorithm.
\end{theorem}

The approach of Theorem~\ref{thm:continuousCaseUpperBound} applies an underlying grid with mesh size $2\sqrt{2}$.
Our algorithm
(1) moves the robots to vertices of the grid,
(2) applies our $\mathcal{O}(1)$-approximation for the discrete case, and
(3) moves the robots from the vertices of the grid to their targets. \new{For a detailed description of the Algorithm of Theorem~\ref{thm:continuousCaseUpperBound} see Section~\ref{subsec:upper_cont}.}

\section{Details for Labeled Grid Permutation}\label{sec:labeledgridperm_details}

\new{In this section, we give the details omitted in the high-level description of our results for the problem variant of labeled grid permutations that we considered in Section~\ref{sec:constappr}.}

\subsection{Details for the NP-Completeness of the Grid Case}
\statement{Theorem}{thm:optimalNPcomplete}
	{\em
	The minimum makespan parallel motion planning problem on a grid is strongly NP-hard.
}
\begin{proof}
	The proof is based on a reduction from the NP-hard problem \textsc{Monotone 3-Sat}, which asks to decide whether a Boolean 3-CNF formula $\varphi$ is satisfiable, where in each clause the literals are either all positive or all negative.
	All coordinates and the makespan are constructed to be polynomial in the input size, implying strong NP-hardness.
	Thus, there is no FPTAS for the problem, unless $\mathrm{P} = \mathrm{NP}$.

	\begin{sloppypar}For the remainder of the proof, let $\varphi$ have $n$ variables $\{x_0,\ldots,x_{n-1}\}$ and $m$ clauses $\{C_1,\ldots,C_m\}$.
	From $\varphi$, we construct an instance of the minimum makespan parallel motion planning problem that has optimum makespan $M$ if $\varphi$ is satisfiable and $M+1$ otherwise.
	During the description of the construction, we keep $M$ variable, fixing its value once the construction is complete.
	The structure of the resulting instance is sketched in Figure~\ref{fig:reduction-overview}.\end{sloppypar}
	\begin{figure}
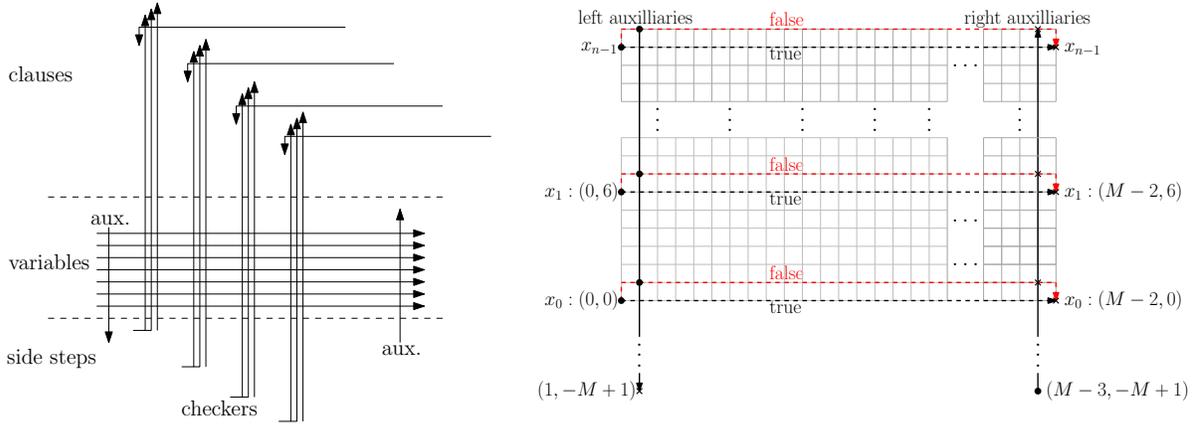

		\begin{center}
			\begin{subfigure}[c]{.45\linewidth}
				\begin{center}
					\resizebox{\linewidth}{!}{\includegraphics{./figures/reduction/overview.pdf}}
				\end{center}
			\end{subfigure}
			\begin{subfigure}[c]{.54\linewidth}
				\begin{center}
					\resizebox{\linewidth}{!}{\includegraphics{./figures/reduction/variable_gadget.pdf}}
				\end{center}
			\end{subfigure}
		\end{center}
		\caption{\emph{Left}: The structure of the resulting parallel motion planning problem instance. \emph{Right}: Start configuration (disks) and target configuration (crosses) of variable robots and their auxiliaries.
			The \emph{left auxiliary robot} for $x_j$ starts at position $(1,6j+1)$ and has to move down towards its target position $(1,6j+1-M)$ in each time step.
			The \emph{right auxiliary robot} for $x_j$ starts at $(M-3,-M+6j+1)$ and has to move up towards its target position $(M-3,6j+1)$.}
		\label{fig:reduction-overview}
		\vspace{-8pt}
	\end{figure}

	Each variable $x_j$ is represented by a \emph{variable robot}.
	Additionally, for each variable there are two \emph{auxiliary robots} that force the variable robot to take one of two different paths to its goal in any solution with makespan $M$, see Figure~\ref{fig:reduction-overview}.
	\revised{The \emph{left auxiliary robots} start at positions $(1,6j+1)$ and move down towards their target positions $(1,6j+1-M)$ in each time step.
	The \emph{right auxiliary robots} start at positions $(M-3,-M+6j+1)$ and have to move up towards their target positions $(M-3,6j+1)$.}
	The variable robot for variable $x_j$ starts at position $(0,6j)$ and has to travel $M-2$ units to the right towards its goal position $(M-2,6j)$.
	In the first time step, each variable robot can either wait or move upwards.
	Afterwards, it must move to the right in every time step until passing the right group of auxiliary robots at $x = M-3$.
	It cannot wait or move down before this point, as this would lead to a collision with the corresponding right auxiliary robot.
	Therefore in any schedule with makespan $M$, after the $k$th time step, each variable robot has $x$-coordinate $k-1$ for any $1 \leq k \leq M-3$.

	For each clause $C_i = \{x_{j_1},x_{j_2},x_{j_3}\}$ with $j_1 < j_2 < j_3$, we have three \emph{checker robots} $c_i^1, c_i^2, c_i^3$ checking whether their corresponding literal satisfies the clause.
	The checkers for clause $C_i$ start at positions $\alpha_i^1 := \big(6(ni+j_1),-6ni-f_i\big),\ldots,\alpha_i^3 := \big(6(ni+j_3),-6ni-f_i\big)$, where $f_i = 1$ iff $C_i$ is negative and $f_i = 0$ otherwise. 
	As depicted in Figure~\ref{fig:reduction-variable-clause}, a checker has to wait one time step for the corresponding variable iff the checked literal is not true.
	\begin{figure}[h]
		\begin{center}
			\resizebox{.75\linewidth}{!}{\includegraphics{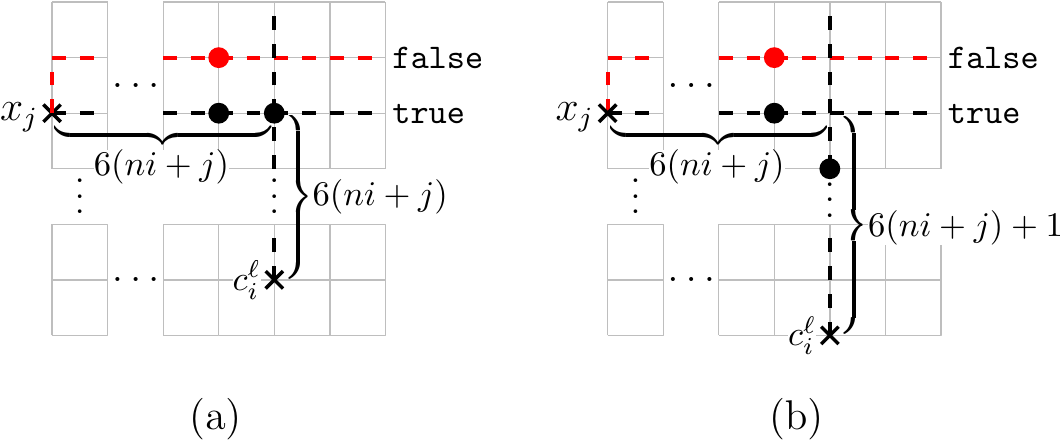}}
		\end{center}
		\vspace{-8pt}
		\caption{(a): A checker $c_i^\ell$ for variable $x_j$ in a positive clause $C_i$. (b): A checker $c_i^\ell$ for variable $x_j$ in a negative clause $C_i$. Checkers must wait iff the variable assignment does not match.}
		\label{fig:reduction-variable-clause}
	\end{figure}

	Checker $c_i^3$ has to move $M-1$ units up to its target position $t_i^3 := \alpha_i^3 + (0,M-1)$.
	Let $d_1 := 6(j_3-j_1)$ be the horizontal distance between the initial positions of $c_i^1$ and $c_i^3$, and let $d_2 := 6(j_3-j_2)$ analogously.
	Both $d_1$ and $d_2$ are always even and at least six; therefore $s_1 := \frac{d_1}{2}+2 < d_1$ and $s_2 := \frac{d_2}{2}+1 < d_2$ are integer.
	We force $c_i^1$ to take $s_1$ steps to the right towards its target position $t_i^1 := \alpha_i^1 + (s_1, M-1-s_1)$.
	Analogously, $c_i^2$ has target position $t_i^2 := \alpha_i^2 + (s_2, M-1-s_2)$.
	Each checker travels a total distance of $M-1$; thus they are allowed to wait for one time step, but have to move on an $xy$-monotone path towards their target position.

	Because moves to the right do not change the position of a checker relative to the variables, we may assume the checkers to move to the right from their initial position before moving up.
	In fact, we enforce this behavior using auxiliary robots as depicted in Figure~\ref{fig:reduction-sidesteps}.
	\begin{figure}[h]
		\begin{center}
			\begin{subfigure}[c]{.45\linewidth}
				\begin{center}
					\resizebox{\linewidth}{!}{\includegraphics{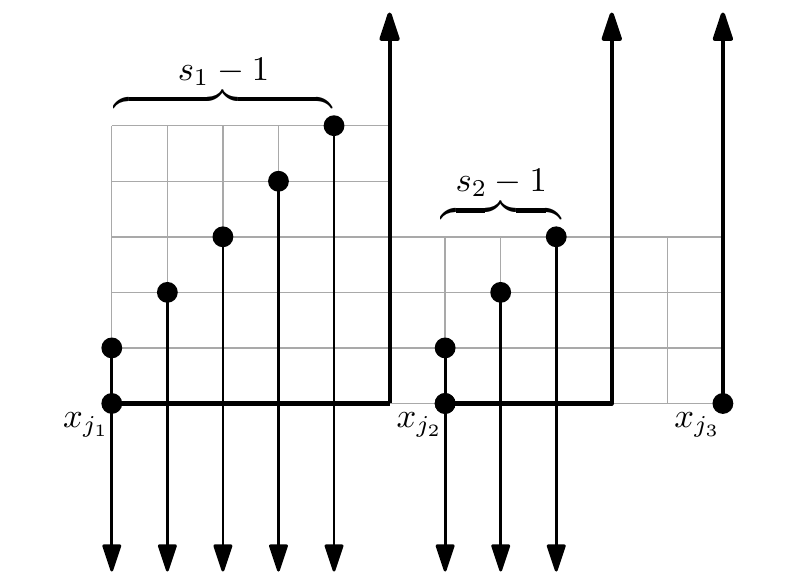}}
				\end{center}
			\end{subfigure}
			\hfill
			\begin{subfigure}[c]{.5\linewidth}
				\begin{center}
					\resizebox{\linewidth}{!}{\includegraphics{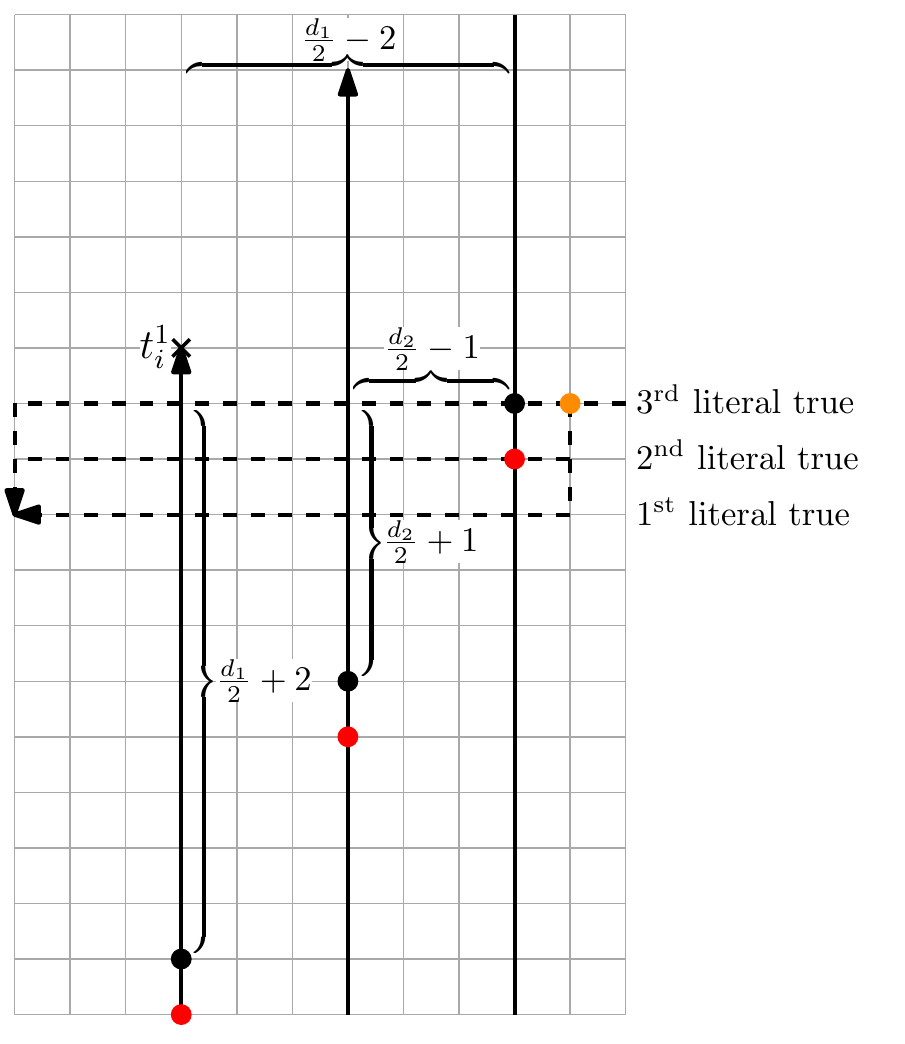}}
				\end{center}
			\end{subfigure}
		\end{center}
		\caption{\emph{Left}: A group of auxiliary robots is used to force the first two checkers of each clause to perform their side steps before moving up. Each auxiliary robot has to move downwards $M$ units.
		\emph{Right}: A clause robot (orange) meeting the corresponding checkers (black for satisfied checkers, red for non-satisfied checkers).}
		\label{fig:reduction-sidesteps}
	\end{figure}
	Moreover, each clause $C_i$ also has a \emph{clause robot} ensuring that there is at least one satisfied literal.
	The clause robots start to the right of the checkers and above the variables and have to move $M-2$ units to the left and two units downwards, and therefore have to move towards their target in every round without waiting for the checkers.
	The clause robot of each clause is placed such that checkers for other clauses cannot interfere with its path, see Figure~\ref{fig:reduction-overview}.
	To be more precise, as shown in Figure~\ref{fig:reduction-sidesteps}, the clause robot stops at position $t_i^1 - (3,3)$ and starts at position $t_i^1 + (-1,M-5)$.
	The vertical offset between the checkers introduced by the side steps that $c_i^1$ and $c_i^2$ perform is chosen such that the clause robot can pass through the checkers without waiting iff one of the checkers did not wait.
	This is the case iff at least one literal of the clause is satisfied.
	
	It remains to determine the critical makespan $M$.
	This critical makespan $M$ must be large enough to allow the checkers of the last clause $C_m$ to pass through the variable robots and their clause robot.
	Moreover, it must also allow the variable robots to cross paths with all checkers.
	The checkers of the last variable travel left of the line $x = 6n(m+1)-6$.
	Therefore, a makespan $M \geq 6n(m+1)$ suffices for the variable robots.
	Regarding the clauses, if the last clause is negative, the starting points of its checkers are located on the line $y = -6nm-1$.
	The topmost variable robot travels below the line $y = 6(n-1) + 1$.
	To keep our argument simple, we want to make sure that the clauses stay strictly above all variables.
	Due to the position of the clauses, this means that we have to ensure that the checker for the first literal of the last clause has target position above the line $y = 6(n-1)+5$.
	Therefore, not accounting for the side steps of the checkers, we have to set $M \geq (6nm+1) + (6(n-1)+5) = 6n(m+1)$.
	Clearly, the number of side steps performed by each checker is less than $6n$.
	Therefore, in total, a critical makespan of $M := 6n(m+2)$ is sufficient.

	In our construction, a makespan of $M$ is feasible iff for every clause robot there is one checker that does not wait, which implies that each clause has a satisfied literal under the assignment induced by the variable robots.
	Therefore, a makespan of $M$ is feasible iff $\varphi$ is satisfiable.

	Finally, observe that even though our reduction uses individually labeled robots, three colors are already sufficient.
	One can use color $1$ for variables, color $2$ for checkers and color $3$ for clauses and all auxiliaries.
\end{proof}

\subsection{Details on Computing a Schedule With a Makespan of $\mathcal{O}(n_1+n_2)$}\label{sec:scheduleNaive}

\new{Next, we give the details of an algorithm that computes a sequence of $\mathcal{O}(n_1+n_2)$ steps transforming an arbitrary start configuration $C_s$ into an arbitrary target configuration $C_t$ of an $n_1 \times n_2$ rectangle, see Lemma~\ref{lem:naive}. This algorithm is based on a sorting algorithm, called {\sc RotateSort} that uses swap operations, in which two robots exchanging their positions within one single step, as elementary operations. As our model does not allow swap operations, we first have to show how to simulate swap operations at the expense of increasing the makespan by a factor upper-bounded by some constant.}

\new{In order to simulate swap operations, we first observe that} LaValle and Yu~\cite{yl-omppgcaeh-16} proved that for a $3 \times 3$-square, each start configuration can be transformed into an arbitrary target configuration. 
This result is easily established for $2 \times 3$-rectangles; 
see Figure~\ref{fi:swap} for how to realize a transposition. 

	\begin{figure}[h]
		\begin{center}
			\resizebox{.41\textwidth}{!}{\includegraphics{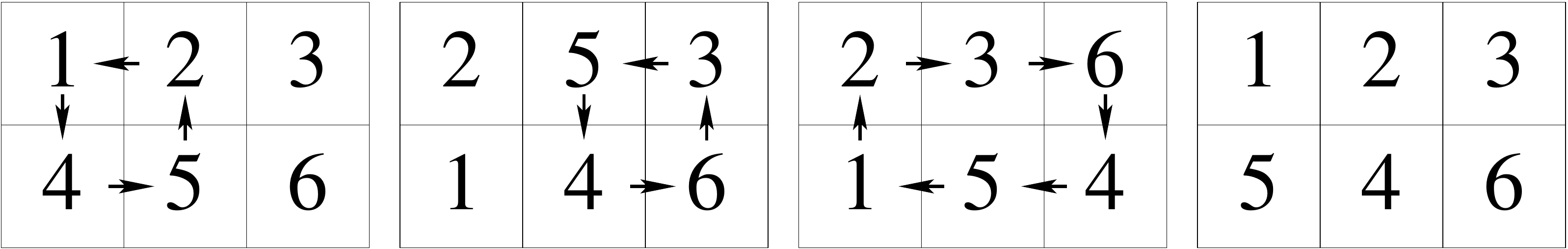}}
		\end{center}
\vspace*{-6pt}
		\caption{Using three moves for swapping two positions in a $2 \times 3$-arrangement.}
		\label{fi:swap}
\vspace*{-12pt}
	\end{figure}

\begin{lemma}\label{lem:3times3}
	For a pair of start and target configurations $C_s$ and $C_t$ of a $2 \times 3$-rectangle, we can compute a sequence of at most seven steps transforming $C_s$ into $C_t$.
\end{lemma}

Lemma~\ref{lem:3times3} is the building block for 
permuting $n_1 \times n_2$ rectangles within makespan $\mathcal{O}(n_1+n_2)$.

\statement{Lemma}{lem:naive}
	{\em For a pair of start and target configuration $C_s$ and $C_t$ of an $n_1 \times n_2$-rectangle, we can compute in polynomial time a sequence of $\mathcal{O}(n_1+n_2)$ steps transforming $C_s$ into~$C_t$.
}
\new{
\begin{proof}
	The straightforward proof relies on covering the rectangle by a set of disjoint $2 \times 3$- and $3 \times 2$-rectangles, on which swap operations are performed in parallel, with each swap operation exchanging the position of two adjacent robots.
	We say that two swap operations are {\em disjoint} if all four positions \nnew{of the two swaps} are distinct.
	Although direct swap operations of adjacent robots are not possible, 
Lemma~\ref{lem:3times3} allows us to perform an arbitrary number of pairwise disjoint swap operations within each $2 \times 3$-rectangle
 with $\mathcal{O}(1)$ transformation steps.
As illustrated in Figure~\ref{fig:2times3covering}, we cover $P$ by twelve different layers of rectangles, 
such that each pair of adjacent unit squares from $P$ lies in one of the $2 \times 3$-rectangles or in one of the $3 \times 2$-rectangles. 

In particular, we distinguish between $2 \times 3$- and $3 \times 2$-rectangles
inside the $n_1 \times n_2$-rectangle. Furthermore, we distinguish between
different positions of $2 \times 3$-rectangles w.r.t.~line numbers modulo $2$
and w.r.t.~column numbers modulo $3$; see
Figures~\ref{fig:2times3covering}a)-f). Analogously, we distinguish between different
positions of $3 \times 2$-rectangles w.r.t.~line numbers modulo~$3$ and w.r.t.~column numbers modulo $2$; see the Figures~\ref{fig:2times3covering}g)-l).
This results in twelve different classes of rectangles. 
	
	\begin{figure}[h]
		\begin{center}
			\resizebox{0.6\textwidth}{!}{\includegraphics{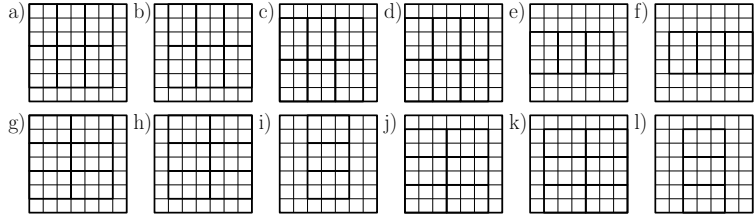}}
		\end{center}
		\caption{Covering of $P$ by pairwise disjoint $2 \times 3$- and $3 \times 2$-rectangles in twelve layers.}
		\label{fig:2times3covering}
	\end{figure}

	Given a set $S$ of pairwise disjoint swap operations, we subdivide $S$ into these twelve layers, 
        such that the two robots of each swap operation lie in the same small rectangle of the corresponding layer.
	Lemma~\ref{lem:3times3} implies that all swap operations of one layer can be done in parallel with $\mathcal{O}(1)$ transformation steps.
	Therefore, all swap operations in $S$ can be done in $\mathcal{O}(1)$ transformation steps.

	This allows us to apply a sorting algorithm for $n_1 \times n_2$-meshes, called {\sc Rotatesort}~\cite{marberg:sorting}, whose only elementary steps are swap operations of adjacent cells.
	We employ {\sc Rotatesort} by labeling the robots in the target configuration based on the snake-like ordering guaranteed by {\sc Rotatesort}.
	Applying {\sc Rotatesort} to the start configuration with the robots labeled in this way, we obtain the required target configuration.
	Marberg and Gafni~\cite{marberg:sorting} show that {\sc Rotatesort} needs $\mathcal{O}(n_1+n_2)$ phases, where each phase consists of pairwise disjoint swap operations.
	This leads to $\mathcal{O}(n_1+n_2)$ transformation steps in our model.
\end{proof}
}

Figure~\ref{fig:tiling}, shows an example of a pair of start and target configuration and the resulting flow.

\begin{figure}[h]
\begin{center}
	\resizebox{.4\textwidth}{!}{\includegraphics{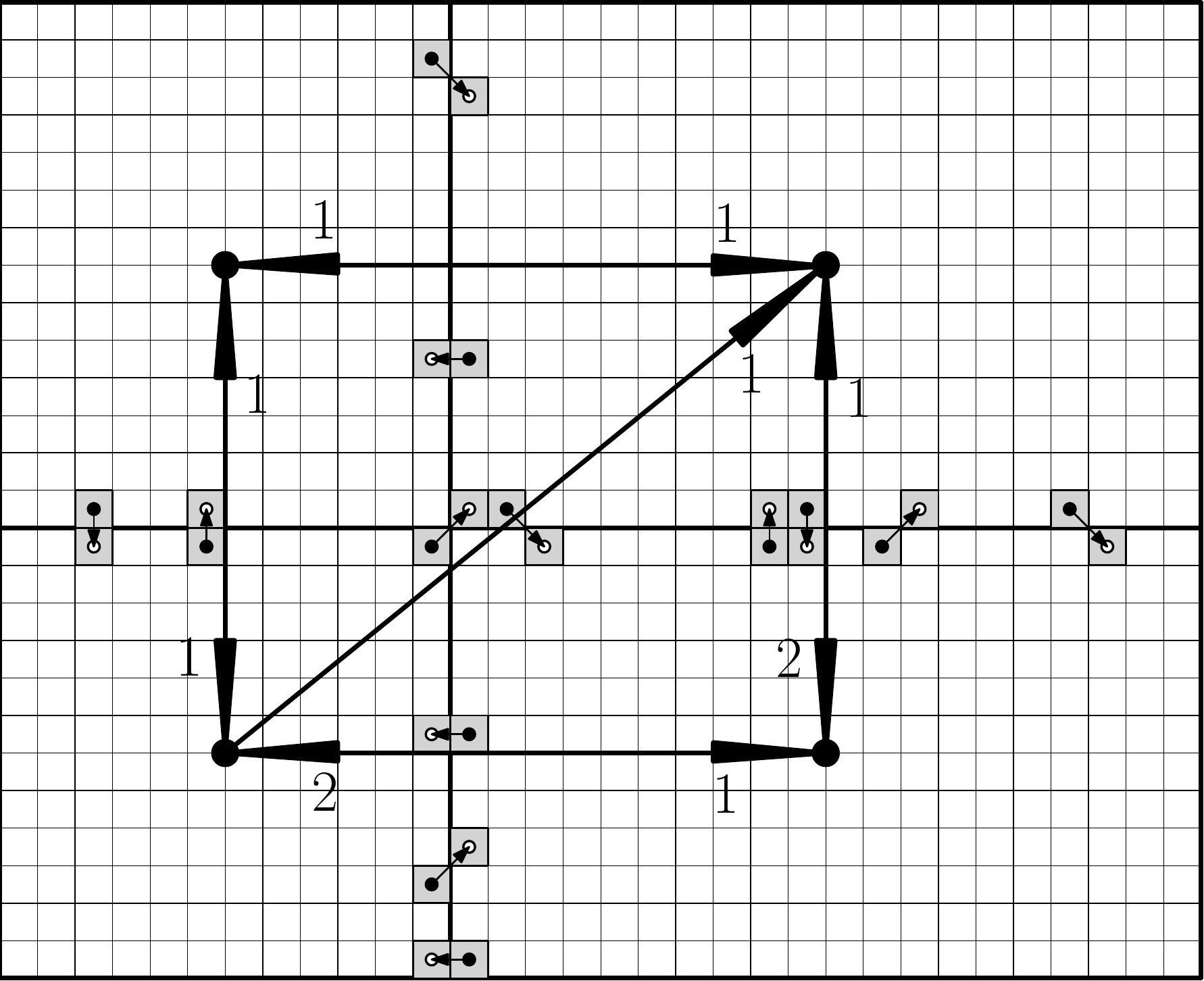}}
\end{center}
\vspace*{-6pt}
\caption{A tiling of an $26 \times 32$-rectangle by four tiles with $d=1$.
Robots not in their target tile are illustrated by small dots.
Their target positions are depicted as white disks.
The dual graph of the tiling is illustrated by large dots and directed edges between them.
The edges of the dual graph are annotated with the value of the flow on the corresponding edge.
In general, it is not guaranteed that robots that have to change the tile lie adjacent to the border between their start and target tile.
However, this is the case for $d=1$, as illustrated in this figure.}
\label{fig:tiling}
\vspace*{-12pt}
\end{figure}

\subsection{Details on the Approach of Theorem~\ref{thm:main2}}\label{sec:polymain2}

\new{
\statement{Theorem}{thm:main2}
	{\em There is an algorithm with running time $\mathcal{O}(N^5)$ that, given an arbitrary pair of start and target configurations of a rectangle $P$ with $N$ robots to be moved and maximum distance $d$ between any start and target position, computes a schedule of makespan $\mathcal{O}(d)$, i.e., an approximation algorithm with constant stretch.
}
\begin{proof} Our algorithm considers the two cases (1) $N \leq \lceil \frac{n_1}{4} \rceil, d$ and (2) $N > \lceil \frac{n_1}{4} \rceil$ or $N>d$ separately as follows:

	In case (1), we apply the following approach whose steps, described next, are all realizable because $N \leq \lceil \frac{n_1}{4} \rceil,d$. We assume w.l.o.g. that $n_1$ and $n_2$ are even. Otherwise, starting from the start configuration, we move all robots from the last line into the second-to-last line and all robots from the last column into the second-to-last column within $\mathcal{O}(d)$ transformation steps. The reversed argument implies that there is a sequence of $\mathcal{O}(d)$ transformation steps leading from an even-sized configuration to the target configuration. Thus, from now on, we restrict our considerations to even-sized rectangles.
	
	For each pair of start and target configurations~$C_s$ and $C_t$ of $P$, there are two configurations~$C_o$ and $C_e$, such that the \nnew{two following conditions are fulfilled: (1) The} coordinates of the robots in $C_o$ are odd and the coordinates of the robots in~$C_e$ are even and (2) $C_s$ and $C_e$ can be transformed into $C_o$ and $C_t$ within $\mathcal{O}(d)$ transformation steps. Thus, we still have to give an approach for how $C_o$ can be transformed into $C_e$ within $\mathcal{O}(d)$ transformation steps.
	
	First of all, we ensure in parallel for all robots that they achieve the position that is induced by the $x$-coordinate of their position in $C_e$ and the $y$-coordinate of their position in $C_o$. We call the corresponding configuration \emph{intermediate configuration $C_i$} with \emph{intermediate positions and coordinates}. In order to obtain the intermediate configuration, starting from $C_o$, we first push in parallel all robots, that have to move to the right, one position upwards, then move them simultaneously to the right until they achieve their intermediate $x$-coordinate, and, push a robot immediately one position downwards when it reaches its intermediate $x$-coordinate, see Figure~\ref{fig:achievingIntermediate}.
	
\begin{figure}[h]
\begin{center}
	\includegraphics[scale=0.4]{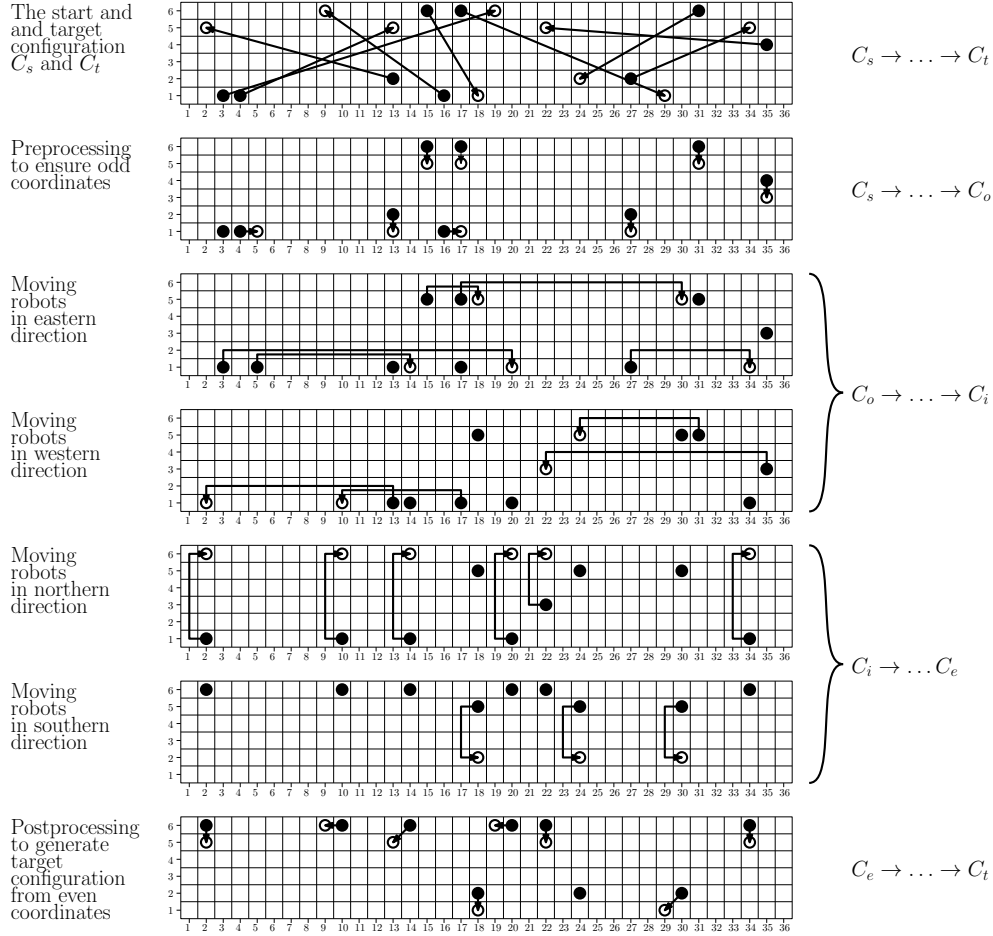}
\end{center}
\vspace*{-6pt}
\caption{A stepwise illustration of the approach for case (1) of Theorem~\ref{thm:main2}.}
\label{fig:achievingIntermediate}
\end{figure}

\begin{figure}[h]
\begin{center}
	\includegraphics[scale=0.5]{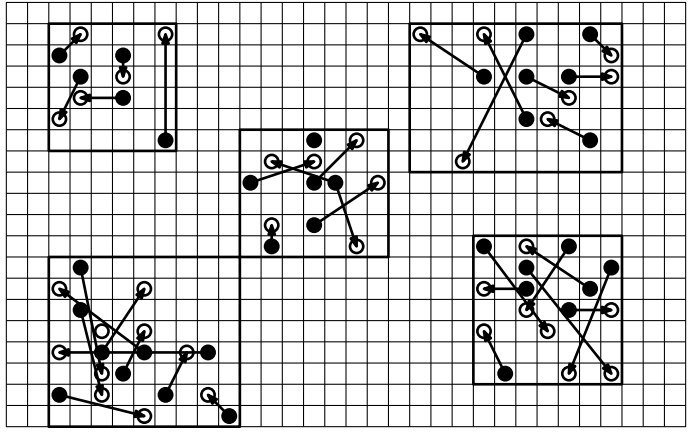}
\end{center}
\vspace*{-6pt}
\caption{An illustration how the approach for case (2) of Theorem~\ref{thm:main2} clusters the pairs of robots' start and target positions.}
\label{fig:polyTimeSecondApproach}
\end{figure}
	
	After that, we apply the analogous approach for robots that have to move to the left, resulting in the intermediate configuration.
	
	Secondly, starting from the intermediate configuration, we apply the above described two-stepped approach for horizontal movements in an analogous version in order to ensure that the $y$-coordinate of each robot $r$ is equal to the $y$-coordinate of $r$ in the configuration $C_e$, while guaranteeing that the $x$-coordinate of $r$ stays the same. This results in the configuration~$C_e$.
	
	The transformation steps leading from $C_s$ to $C_o$ and leading from $C_e$ to $C_t$ can be computed in $\mathcal{O}(N \cdot N)$ time by making use of the fact that the robots' positions are explicitly given via their coordinates. The same reasoning implies that the sequences of transformation steps leading from $C_o$ to the intermediate configuration and leading from the intermediate configuration to $C_e$ can be computed in $\mathcal{O}(N \cdot N)$ time.
	
	In case (2), we apply the approach of Theorem~\ref{thm:main} as a subroutine in the following approach: For each robot we consider the smallest rectangle that contains the robot's start and target positions. If the rectangle has a height or width of $1$, we extend the height or width to $2$. Now we iteratively replace two rectangles $R_1$ and $R_2$ intersecting each other by the smallest rectangle that contains $R_1$ and $R_2$. 
	
	This results in a set of rectangles that are pairwise intersection free, allowing us to apply the approach of Theorem~\ref{thm:main} to each resulting rectangle in parallel, while ensuring that each robot is involved in at most one application of the approach of Theorem~\ref{thm:main}.
	
	As the side lengths of the initial rectangles are upper-bounded by $d$, we conclude that the sum of the lengths of the finally computed rectangles is upper bounded by $N \cdot d$, which in turn is upper bounded by $N^2$ in that case. This implies a running time of $\mathcal{O}(d \cdot N^2) \leq \mathcal{O}(N^3)$.

	\end{proof}
}

\subsection{Details on Step 3: Computing a Flow Partition}\label{sec:computeFlowPartition}

\statement{Lemma}{lem:computePartition}
	{\em
	We can compute a $(d,\mathcal{O}(d))$-partition of $G_T$ in polynomial time.
}
\begin{proof}
In a slight abuse of notation, throughout this proof, the elements in sets of cycles are not necessarily unique.
A $(d,\mathcal{O}(d))$-partition can be constructed using the following steps.
\begin{itemize}
        \item We start by computing a $(1,h)$-partition $\mathbb{C}_{\bigcirc}$ of $G_T$ consisting of $h \leq n_1n_2$ cycles.
        This is possible because $G_T$ is a circulation.
        If a cycle $C$ intersects itself, we subdivide $C$ into smaller cycles that are intersection-free.
        Furthermore, $h$ is clearly upper bounded by the number of robots $n_1n_2$, because every robot can contribute only $1$ to the sum of all edges in $G_T$.
        As the cycles do not self-intersect, we can partition the cycles $\mathbb{C}_{\bigcirc}$ by their orientation, obtaining the set $\mathbb{C}_{\circlearrowright}$ of clockwise and the set $\mathbb{C}_{\circlearrowleft}$ of counterclockwise cycles.

        \item We use $\mathbb{C}_{\circlearrowright}$ and $\mathbb{C}_{\circlearrowleft}$ to compute a $(1,h')$-partition $\mathbb{C}_{\circlearrowright}^1 \cup \mathbb{C}_{\circlearrowright}^2 \cup \mathbb{C}_{\circlearrowleft}^1 \cup \mathbb{C}_{\circlearrowleft}^2$ with $h' \leq n_1n_2$, such that two cycles from the same subset $\mathbb{C}_{\circlearrowright}^1$,  $\mathbb{C}_{\circlearrowright}^2$, $\mathbb{C}_{\circlearrowleft}^1$, or $\mathbb{C}_{\circlearrowleft}^2$ share a common orientation.
        Furthermore, we guarantee that two cycles from the same subset are either edge-disjoint or one lies nested in the other.
        A partition such as this can be constructed by applying a recursive peeling algorithm to $\mathbb{C}_{\circlearrowright}$ and $\mathbb{C}_{\circlearrowleft}$ as depicted in Figure~\ref{fig:intersectionFree}, yielding a decomposition of the flow induced by $\mathbb{C}_{\circlearrowright}$ into two cycle sets $\mathbb{C}_{\circlearrowright}^1$ and $\mathbb{C}_{\circlearrowright}^2$, where $\mathbb{C}_{\circlearrowright}^1$ consists of clockwise cycles and $\mathbb{C}_{\circlearrowright}^2$ consists of counterclockwise cycles, and a similar partition of $\mathbb{C}_{\circlearrowleft}$. \new{In particular, we apply the following approach iteratively to $\mathbb{C}_{\circlearrowright}$: We consider the union $A$ of the area bounded by the cycles from $\mathbb{C}_{\circlearrowright}$. We remove a flow value of $1$ from all edges of the outer boundary component of $A$. In particular, we add the corresponding $1$-subflow $G_1$ to $\mathbb{C}_{\circlearrowright}^1$ and remove $G_1$ from $\mathbb{C}_{\circlearrowright}$. Analogously, we remove $1$-subflows from $\mathbb{C}_{\circlearrowright}$ that are induced by inner boundary components and add these $1$-subflows to $\mathbb{C}_{\circlearrowright}^2$.}

        \begin{figure}[h]
                \begin{center}
                        \resizebox{.75\linewidth}{!}{\includegraphics{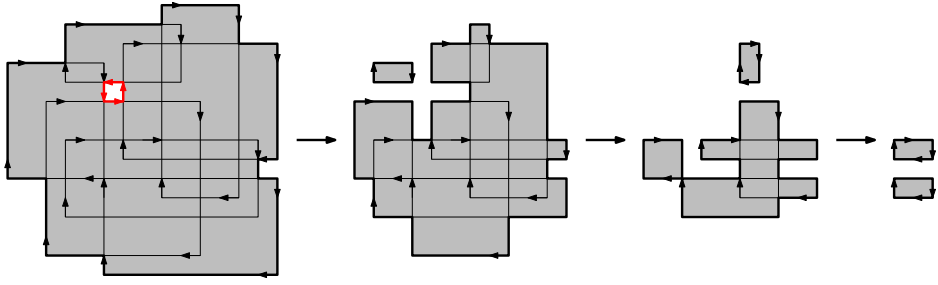}}
                \end{center}
                \caption{Recursive peeling of the area bounded by the cycles from $\mathbb{C}_{\circlearrowright}$, resulting in clockwise cycles (thick black cycles).
                        Cycles constituting the boundary of \emph{holes} are counterclockwise (thick red cycles).
                        Note that an edge $e$ vanishes when $f_T(e)$ cycles containing that edge are removed by the peeling algorithm described above.}
        \end{figure}

        \item Afterwards, we partition each set $\mathbb{C}_{\circlearrowright}^1$, $\mathbb{C}_{\circlearrowright}^2$, $\mathbb{C}_{\circlearrowleft}^1$, and $\mathbb{C}_{\circlearrowleft}^2$ into $\mathcal{O}(d)$ subsets, each inducing a $d$-subflow of $G_T$.
        This can be done as follows.
        Let $\mathbb{C} \in \{ \mathbb{C}_{\circlearrowright}^1, \mathbb{C}_{\circlearrowright}^2, \mathbb{C}_{\circlearrowleft}^1, \mathbb{C}_{\circlearrowleft}^2 \}$.
        Recall that every pair of cycles from $\mathbb{C}$ either consists of one cycle nested inside the other or of edge-disjoint cycles.
        The cycles induce a dual forest $D=(\mathbb{C},E_D)$, where a cycle $v$ has a child $w$ iff $w$ lies inside $v$ and there is no other cycle lying in $v$ that $w$ lies in.
        We label the cycles by their depth in $D$ modulo $576d$ and let $G_i$ be the flow induced by all cycles carrying label $i$, thus obtaining $\mathcal{O}(d)$ subflows $G_i$.
        \end{itemize}

 	\noindent Finally, we show that each subflow $G_i$ obtained in this way is a $d$-subflow of $G_T$.
        To this end, we observe the following.
        Let $e \in E_T$ be an arbitrarily chosen edge and let $v,w \in \mathbb{C}$ be two cycles sharing $e$.
        This implies that $v$ and $w$ lie nested inside of each other; w.l.o.g., assume that $w$ lies inside $v$.
        Thus, in $D$, $v$ lies on the path from $w$ to its root, and $e$ is contained in all cycles on the path between $v$ and $w$.
        On the other hand, due to Observation~\ref{obs:maxEdgeWeight}, all cycles containing $e$ lie on a path of length at most $576d^2$ in $D$.
        Therefore, $e$ has a weight of at most $\frac{576d^2}{576d} = d$ in each $G_i$, and $G_i$ is a $d$-subflow.%
 \end{proof}

\subsection{Details on a Subroutine of Step 4: Realizing a Single Subflow}\label{sec:RealizSingleSubflow}

\new{By Lemma~\ref{lem:transStep}, we give an approach that computes a schedule of constant length for a given $d$-subflow.}

\statement{Lemma}{lem:transStep}
	{\em Let $G'_T=(T,E_T',f_T')$ be a planar unidirectional $d$-subflow.
	There is a polynomial-time algorithm that computes a \textcolor{black}{schedule $C_1 \rightarrow \dots \rightarrow C_{k+1}$ realizing $G'_T$ for a constant $k \in \mathcal{O}(1)$}.
}
\begin{proof}
Our algorithm uses $k = \mathcal{O}(d)$ preprocessing steps $C_1 \rightarrow \dots \rightarrow C_k$, as depicted in Figure~\ref{fig:transStepA}(a)+(b), and one final realization step $C_k \rightarrow C_{k+1}$, shown in Figure~\ref{fig:transStepA}(c), moving the robots from their start tiles into their target tiles.
The preprocessing \new{replaces} diagonal edges \new{by pairs of orthogonal edges, see the red arrows in Figure~\ref{fig:transStepA}(a),} and places the moving robots next to the border of their target tiles. \new{Note that the replacements of the diagonal edges cannot be done as part of the preprocessing of Step 2. \nnew{This is} because the replaced diagonal edges may be part of circular flows that cannot be realized locally, as it is done for crossing or bidirectional edges in Step~2 of our algorithm.}

\new{For the final realization step we compute a pairwise disjoint matching
between incoming and outgoing robots, such that each pair is connected by a
tunnel inside the corresponding tile in which these tunnels do not intersect each
other, see Figure~\ref{fig:transStepA}(a). The final realization step is given
via the robots' motion induced by moving each robot into the interior of the
tile and by moving this one-step motion through the corresponding tunnel into
the direction of the corresponding outgoing robot.}

        \noindent \textbf{The preprocessing steps $C_1 \rightarrow \dots \rightarrow C_{k}$:}
        Let $v$ be an arbitrary tile.
        We place all robots corresponding to horizontal and vertical edges $(v,w)$ of $G_T'$ in a row adjacent to the side shared by $v$ and $w$.
        We can do this for all tiles using $\mathcal{O}(d)$ parallel steps by applying Lemma~\ref{lem:naive}.

        Next, we eliminate diagonal edges $(w,v) \in E_T'$ as follows.
        There are two tiles sharing a side with both $w$ and $v$; let $u$ be one of them.
        First we place the $f_T'((w,v))$ robots with start tile $w$ and target tile $v$ in a row next to the side between $w$ and $u$.
        Then, we move them to $u$ by exchanging them with $f_T'((w,v))$ robots with start and target tile $u$ that lie next to the side between $u$ and $v$, as shown in Figure~\ref{fig:transStepA}(a).
        In the resulting flow, the diagonal edge $(w,v)$ with weight $f_T'((w,v))$ is replaced by adding a flow of value $f_T'((w,v))$ on the edges $(w,u),(u,v)$.

        We process all tiles as described above in two parallel phases by applying Lemma~\ref{lem:naive} twice: first on all rows with even index and then on all rows with odd index, thus ensuring that parallel applications of Lemma~\ref{lem:naive} do not interfere with each other.

        \noindent \textbf{The realization step $C_{k} \rightarrow C_{k+1}$:}
        Let $t$ be an arbitrary tile.
        For the transformation step $C_{k} \rightarrow C_{k+1}$, we need a matching between incoming and outgoing robots of $t$, such that there is a set of non-intersecting paths in $t$ connecting each incoming robot with its corresponding outgoing robot.
        As illustrated in Figure~\ref{fig:transStepA}(c), these paths induce the required transformation $C_{k} \rightarrow C_{k+1}$.

        We compute this matching by selecting an incoming robot $r_{in}$ and matching it to a robot $r_{out}$, such that there is a path $p \subseteq \partial t$ between $r_{in}$ and $r_{out}$ that does not touch another incoming or outgoing robot.
        We remove the matched robots from consideration and repeat the matching procedure until no further unmatched robots exist.

        The non-intersecting paths between the positions of the matched robots are constructed as follows.
        For $i \geq 1$, the $i$th \emph{hull} of~$t$ is the union of all squares on the boundary of the rectangle remaining after the hulls $1,\dots,i-1$ are removed.
        The path between $r_{in}$ and $r_{out}$ consists of three pieces, as shown in Figure~\ref{fig:transStepA}(c).
        For the $i$th matched pair of robots, the initial and the last part of the path are straight line segments orthogonal to $\partial t$, from the position of $r_{in}$ to the $d+i$th hull and from the $d+i$th hull towards the position of $r_{out}$.
        The main part of the path lies on the $d+i$th hull, connecting the end of the initial part to the beginning of the last part.%
\end{proof}

\subsection{Details on Step 4: Realizing All Subflows}\label{sec:realizAllSubflows}

\statement{Lemma}{lem:sequenceTransStep}
{\em
Let $\mathcal{S} := \langle G_1=(V_1,E_1,f_1),\dots,G_{\ell}=(V_{\ell},E_{\ell},f_{\ell}) \rangle$ be a sequence of $\ell \leq d$ unidirectional planar $d$-subflows of $G_T$.
	There is a polynomial-time algorithm computing $\mathcal{O}(d)  + \ell$ transformation steps $C_1 \rightarrow \dots \rightarrow C_{k+\ell}$ realizing $\mathcal{S}$.
}
\begin{proof}
        Let $t$ be an arbitrary tile.
        Similar to the approach of Lemma~\ref{lem:transStep}, we first apply a preprocessing \textcolor{black}{step} guaranteeing that the robots to be moved into or out of $t$ are in the right position close to the boundary of $t$.
        Thereafter we move the robots into their target tiles, using $\ell$ applications of the algorithm from Lemma~\ref{lem:transStep} without the preprocessing phase. \new{In particular, we realize a sequence of $\ell$ $d$-subflows by applying $\ell$ times the single realization step of the algorithm from Lemma~\ref{lem:transStep}.}
        
        \new{In order to ensure that a sequence of $\ell$ realization steps from Lemma~\ref{lem:transStep} without intermediate preprocessing steps realizes a sequence of $\ell$ $d$-subflows, we apply the following $\mathcal{O}(d)$ preprocessing steps for all $\ell$ realization steps in advance: For each side of the tile $t$, we place all leaving or entering robots that belong to the same subflow in a common row and stack these rows in the order which is induced by the sequence of the subflows to be realized, see Figure~\ref{fig:sequenceSubflowsA}(a). Finally, pushing all stacked robots downwards into the direction of the boundary $\partial t$ of the tile ensures, that processing one realization step implies that all robots involved in the following realization step lie in a row adjacent to $\partial t$, see Figure~\ref{fig:sequenceSubflowsA}.}

\begin{figure}[ht]
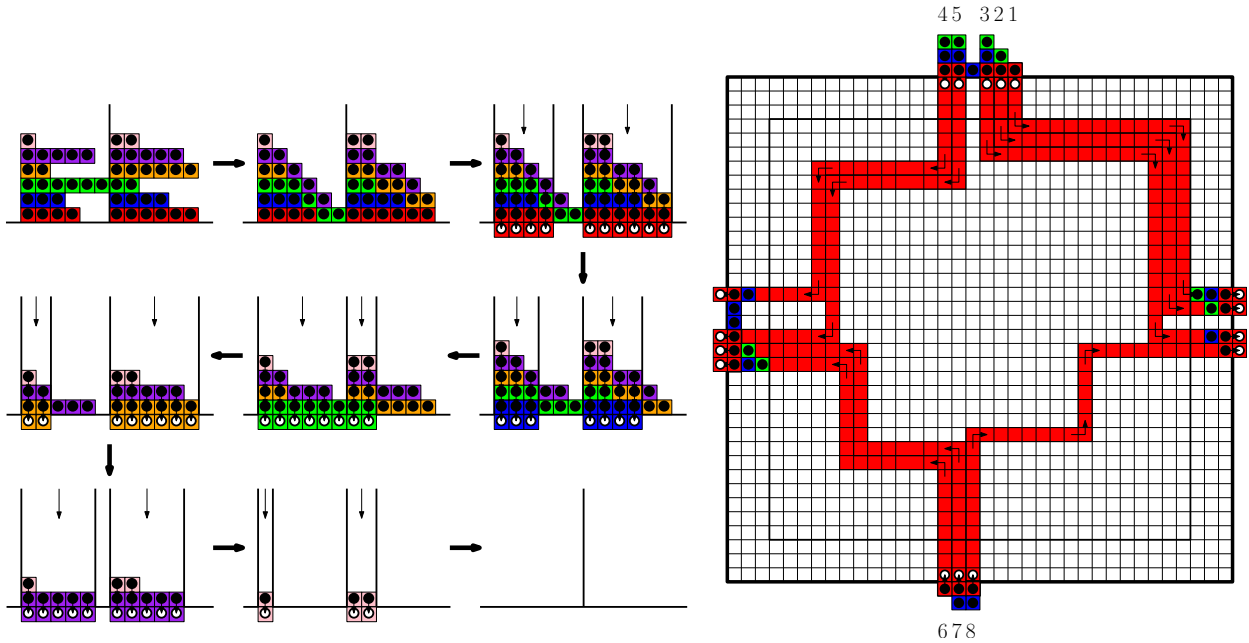

	\begin{center}
		\begin{subfigure}[b]{.55\linewidth}
			\resizebox{\linewidth}{!}{\includegraphics{figures/sequenceSubflows.pdf}}
			\vspace{5pt}
			\caption{Stacking the rows of robots corresponding to the flow values on the edges of the subflows to be realized.}
		\end{subfigure}
		\hfill
		\begin{subfigure}[b]{.43\linewidth}
			\resizebox{\linewidth}{!}{\includegraphics{figures/sequenceSubflows3.pdf}}
			\caption{Our preprocessing step applied to the example of Figure~\ref{fig:transStepA} and the path matching (red) of the first step.}
		\end{subfigure}
	\end{center}
	\vspace*{-12pt}
	\caption{Realizing a sequence of subflows by stacking the rows of robots to be moved onto each other in the order the subflows are realized in.}
	\label{fig:sequenceSubflowsA}
\end{figure}

In the following we describe how we place the robots in their start tiles as a preprocessing step.
First, we use the same preprocessing step as in Lemma~\ref{lem:transStep} to eliminate diagonal edges.
For a simplified illustration, we describe the remainder of the preprocessing in two steps that can be realized by just one application of Lemma~\ref{lem:naive}.
After elimination of diagonal edges, we proceed by stacking the rows of robots moving out of $t$ in the order in which the subflows are to be processed, see Figure~\ref{fig:sequenceSubflowsA}(a).
Then we push the robots towards the boundary $\partial t$ of their start tile until they meet either $\partial t$ or another moving robot.
See Figure~\ref{fig:sequenceSubflowsA}(a),~image 2 for an example.

This preprocessing ensures that, after each application of the algorithm of Lemma~\ref{lem:transStep}, all robots moving out of $t$ in the next transformation step lie in a row adjacent to $\partial t$.
Therefore this preprocessing can be used to replace the preprocessing done in Lemma~\ref{lem:transStep}.
For an example, see Figure~\ref{fig:sequenceSubflowsA}(a), images 3--9.

As $\ell \leq d$, the stacked rows have a height of at most $d$.
Thus, they are contained in hulls $1$ to $d$.
Therefore, and because the flows are unidirectional and diagonals are eliminated, the structure of the stacks is not damaged by the applications of Lemma~\ref{lem:transStep}, allowing us to realize $\ell \leq d$ subflows in $\mathcal{O}(d)$ transformation steps instead of one.
\end{proof}

\subsection{Runtime Analysis \new{of the Algorithm of Theorem~\ref{thm:main}}}\label{sec:togther}

\statement{Theorem}{thm:main}
	{\em
		There is an algorithm with running time $\mathcal{O}(dn_1n_2)$ that, given an arbitrary pair of start and target configurations of an $n_1 \times n_2$-rectangle with
maximum distance $d$ between any start and target position, computes a schedule of makespan $\mathcal{O}(d)$, i.e., an approximation algorithm with constant stretch.
	}

\begin{proof}[Proof of Theorem~\ref{thm:main}]
	The steps of our algorithm have the following time complexity:
	
	\noindent\textbf{Initialization step 1:} Computing $d$, $T$ and $G_T$ is possible in $\mathcal{O}(n_1n_2)$ time.
	 
	\noindent\textbf{Step 2 \& 5:} The application of Lemma~\ref{lem:naive} requires $\mathcal{O}(d^3)$ time for each tile, so these steps can be done in $\mathcal{O}(dn_1n_2)$ time.
	 
	\noindent\textbf{Step 3:} All subroutines of Step 3 can be done in an overall time of $\mathcal{O}(n_1n_2)$. 
		In particular, the $(1,h)$-partition $\mathbb{C}_{\bigcirc}$ of $G_T$ can be computed in $\sum_{e \in E_T}f_T(e) \in \mathcal{O}(n_1n_2)$ time by a simple greedy algorithm.
		The number of edges in all cycles from $\mathbb{C}_{\bigcirc}$ combined is at most $n_1n_2$, which is the number of robots in $P$.
		Thus, resolving self-intersections of cycles in $\mathbb{C}_{\bigcirc}$ can be done in $\mathcal{O}(n_1n_2)$ time.
		As $|\mathbb{C}_{\bigcirc}| \in \mathcal{O}(n_1n_2)$, the partition of $\mathbb{C}_{\bigcirc}$ into $\mathbb{C}_{\circlearrowright}^1$, $\mathbb{C}_{\circlearrowright}^2$, $\mathbb{C}_{\circlearrowleft}^1$, and $\mathbb{C}_{\circlearrowleft}^2$ takes time $\mathcal{O}(n_1n_2)$.
		Furthermore, the partitioning of $\mathbb{C}_{\circlearrowright}^1$, $\mathbb{C}_{\circlearrowright}^2$, $\mathbb{C}_{\circlearrowleft}^1$, and $\mathbb{C}_{\circlearrowleft}^2$ into $\mathcal{O}(d)$ $d$-subflows can be done in time~$\mathcal{O}(n_1n_2)$.
		
	\noindent\textbf{Step 4:} The parallel applications of Lemma~\ref{lem:naive} to disjoint rectangles can be computed in $\mathcal{O}(dn_1n_2)$.
		Furthermore, the construction of all connecting paths between incoming and outgoing robots for all tiles needs $\mathcal{O}(dn_1n_2)$ time per application of the algorithm of Lemma~\ref{lem:sequenceTransStep}.
		By applying Lemma~\ref{lem:sequenceTransStep} constantly many times, Step 4 needs $\mathcal{O}(dn_1n_2)$ time.
\end{proof}

\section{Details for Variants on Labeling}\label{sec:unlabaled_details}

In this section we show how to extend our approach of Section~\ref{sec:constappr}
to 
the \emph{unlabeled} and, more generally, the \emph{colored} variant of the
parallel robot motion-planning problem.

\statement{Theorem}{thm:unlabeled}
	{\em
There is an algorithm with running time $\mathcal{O}(k(N)^{1.5} \log (N) + N^5)$ for computing, given start and target images $I_s,I_t$
with maximum distance~$d$ between start and target positions, 
an $\mathcal{O}(1)$-approximation of the optimal makespan $M$ and a corresponding \textcolor{black}{schedule}.
}

\begin{proof}
	We transform the input into an instance of the labeled variant, such that an $\mathcal{O}(1)$-approximation for the labeled instance provides an $\mathcal{O}(1)$-approximation for the colored instance. 
	For each color $i$, we consider the two point sets $A^i,B^i \subset \mathbb{R}^2$, where $A^i$ contains the center points $a^i_v$ of all unit squares $v \in I^i_s$ and $B^i$ contains the center points $b^i_v$ of all $v \in I^i_t$.

	A {\em bottleneck matching} between $A^i$ and $B^i$ is a perfect matching between $A^i$ and $B^i$ that minimizes the maximal distance.
	The cost of an optimal bottleneck matching between $A^i$ and $B^i$ is in $\mathcal{O}(M)$, because a transformation sequence induces a bottleneck matching on all color classes.
	Efrat et al.~\cite{efrat:geometry} show that the geometric bottleneck matching problem can be solved in $\mathcal{O}(|A+B|^{1.5} \log |A+B|)$ time.

	A set of $k$ bottleneck matchings between the sets $A^i$ and $B^i$ induces labeled start and target configurations $C_s,C_t$.
	Applying the algorithm from Section~\ref{sec:constappr} to these yields a sequence of transformation steps of length $\mathcal{O}(M)$.
\end{proof}

\section{Details for Continuous Motion}\label{sec:continuous_details}

In this section, we consider the continuous geometric case in which the robots 
are identical geometric objects that have to move into a target configuration in the plane without overlapping at any point in time.
We want to minimize the makespan under these conditions, where the velocity of each robot is bounded by 1.

\subsection{A Lower Bound for Unbounded Environments}\label{sec:lowerBoundCont}
In this section we give a worst-case lower bound of $\Omega(N^{1/4}d)$ for the continuous makespan \textcolor{black}{where $N$ is the number of robots}.
To be more precise, we construct a pair of start and target configurations of $N$ robots as illustrated in Figure~\ref{fig:lowerBoundStartAndTarget}(a).
In this instance, we have $d=2$.
In Theorem~\ref{thm:lowerBoundContinuousMakespan}, we show that the optimal continuous makespan of this instance is in $\Omega(N^{1/4})$, yielding the worst-case lower bound stated above.

\begin{figure}[ht]
  \begin{center}
	\begin{subfigure}[b]{.28\linewidth}
		\resizebox{\linewidth}{!}{\includegraphics{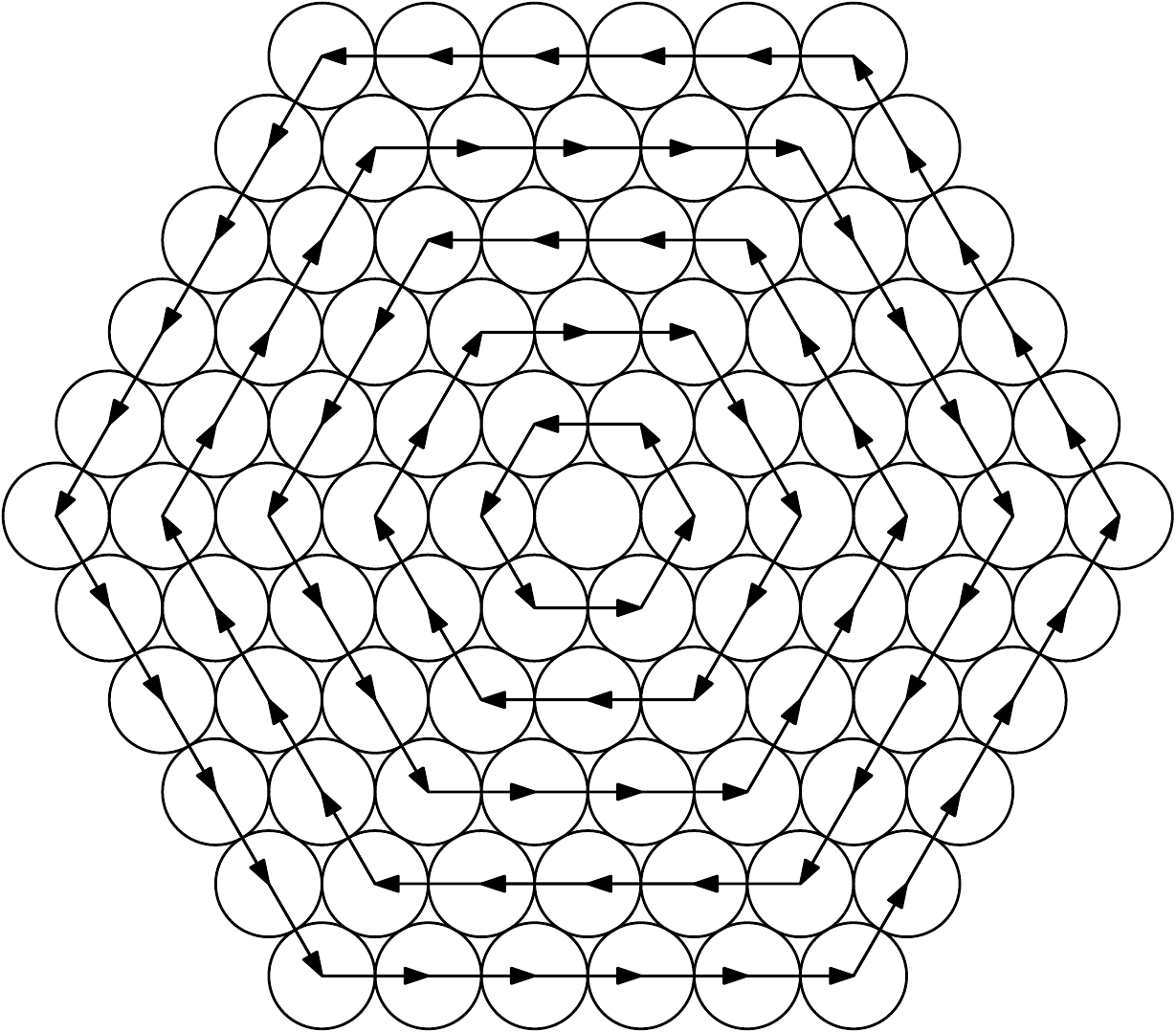}}
		\caption{Start and target positions of the robots.}
	\end{subfigure}
	\hspace{.015\linewidth}
	\begin{subfigure}[b]{.28\linewidth}
		\resizebox{\linewidth}{!}{\includegraphics{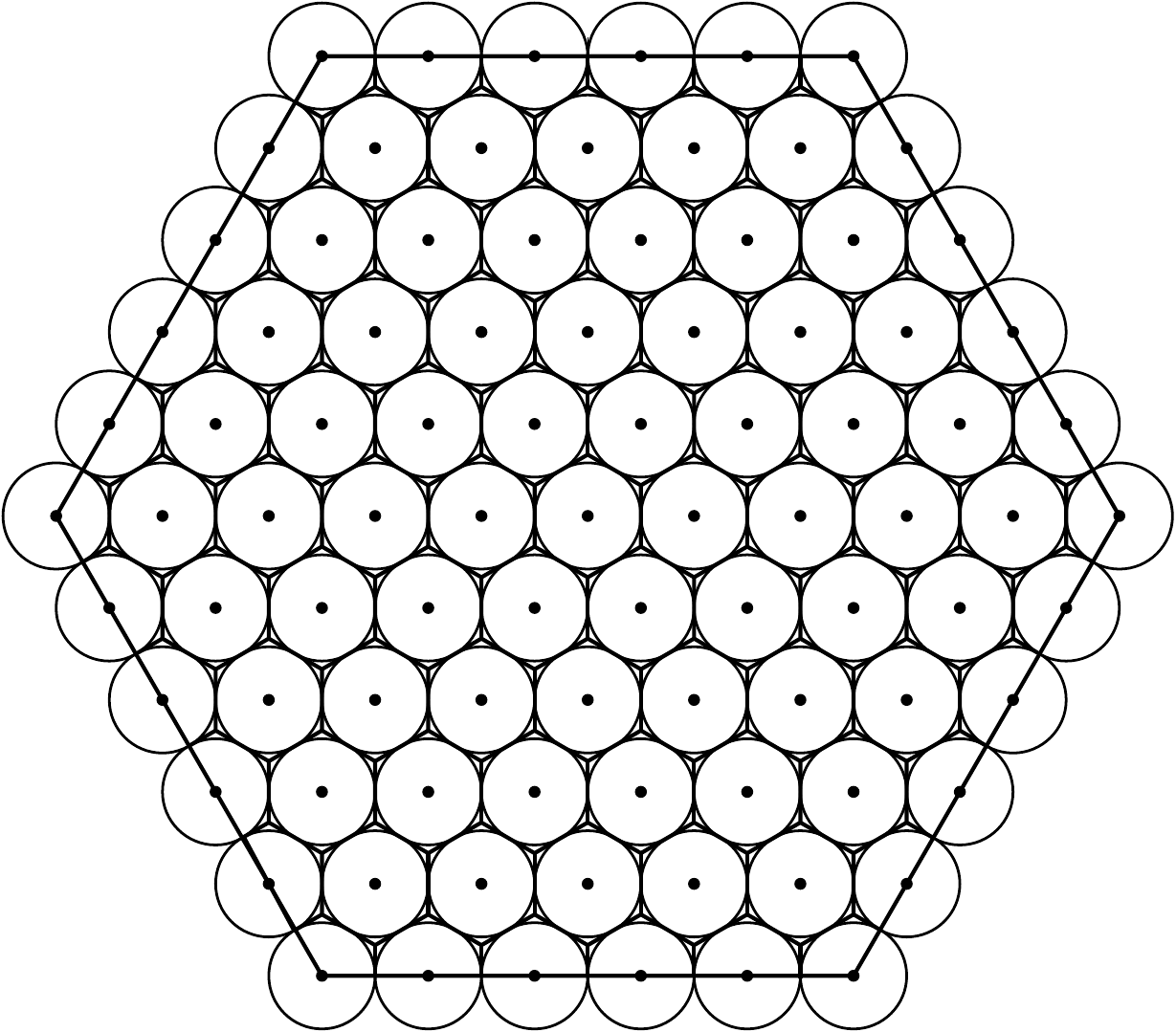}}
		\caption{Voronoi diagram in the start and target configuration.}
	\end{subfigure}
	\hspace{.015\linewidth}
	\begin{subfigure}[b]{.37\linewidth}
		\resizebox{\linewidth}{!}{\includegraphics{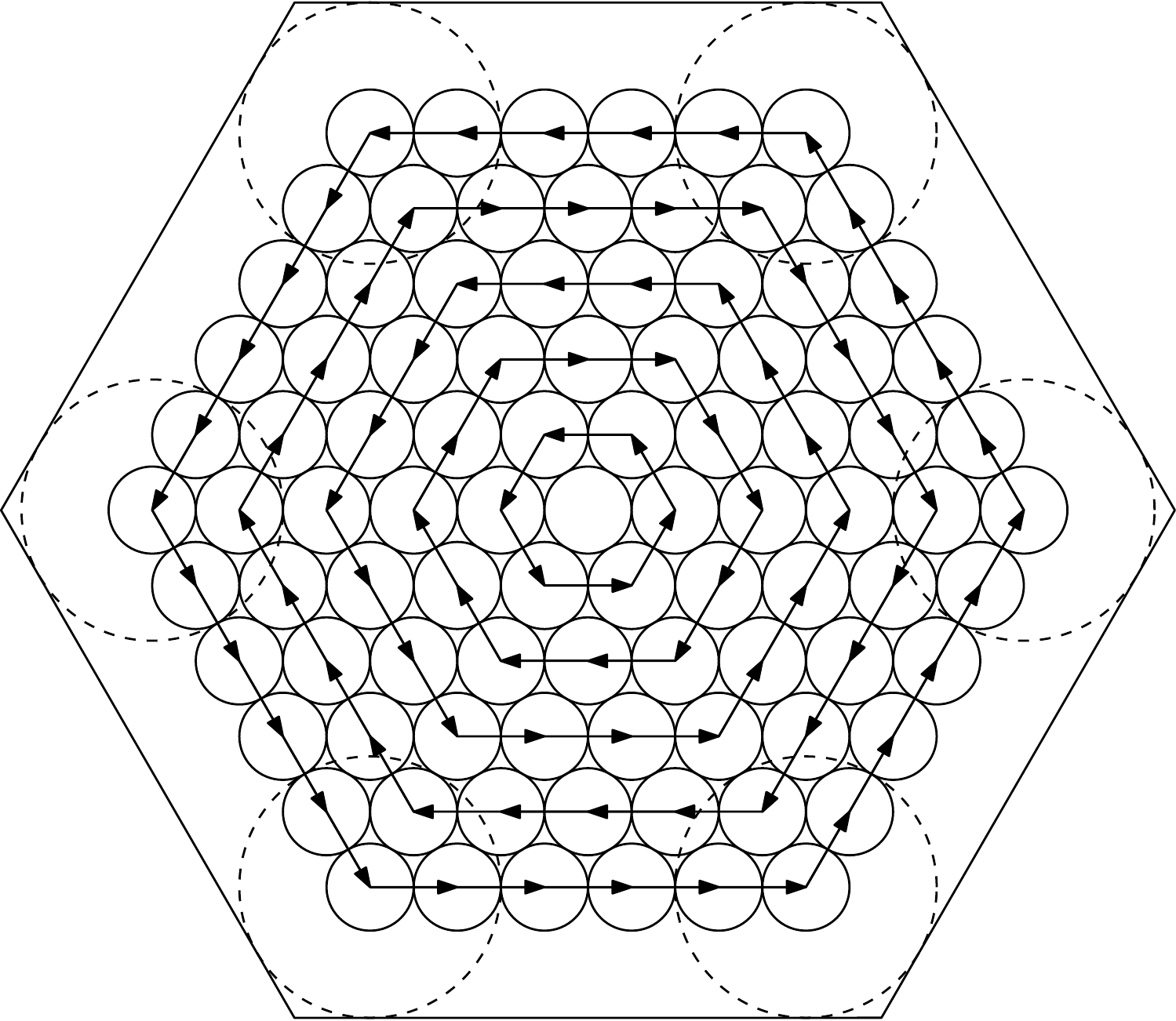}}	
		\caption{Bounding polygon for the moving robots.}
	\end{subfigure}
  \end{center}
  \vspace*{-6pt}
  \caption{The start and target configurations of our lower-bound construction \textcolor{black}{where an arrow points from a start position to the corresponding target position.}}
  \label{fig:lowerBoundStartAndTarget}
\end{figure}

More formally, let $\{ m_1,\dots,m_N \}$ be an arbitrary trajectory set with makespan $M$, realizing the start and target configurations as illustrated in Figure~\ref{fig:lowerBoundStartAndTarget}(a).
By applying a simple continuity argument, we show that there must be a point in time $t \in [0,M]$ such that the area of \textcolor{black}{the convex hull} $\CHull(m_1(t),\dots,m_N(t))$ \textcolor{black}{of $m_1(t),\dots,m_N(t)$} is lower bounded by $cN+\Omega(N^{3/4})$, where $cN$ is the area of the of $\CHull(m_1(0),\dots,m_N(0))$.
Assume $M \in o \left( N^{1/4} \right)$ and consider the area of $\CHull(m_1(t'),\dots,m_N(t'))$ at some point $t' \in [0,M]$.
This area is at most $cN + \mathcal{O}(\sqrt{N}) \cdot o\left(N^{1/4} \right)$, because asymptotically, the area gained during the movement is bounded by the product of makespan and circumference.
This contradicts the lower bound stated above.

	A key ingredient for the construction of the time point $t \in [0,M]$ is the fact that the distance between the centers of two robots change continuously.
	In fact, we know that the Euclidean distance between two centers is \emph{$2$-Lipschitz}, because the velocity of the robots is bounded by 1.
	
	\begin{definition}\label{def:lipschitz}
		A function $f : \mathbb{R} \rightarrow \mathbb{R}$ is $\lambda$-Lipschitz (continuous) if $|f(x)-f(y)| \leq \lambda |x-y|$ holds for all $x,y \in \mathbb{R}$.
	\end{definition}
	\begin{observation}\label{obs:distLipschitz}
		For all $i,j \in R$, the distance between the centers $m_i(\cdot)$ and $m_j(\cdot)$ of robots $i$ and $j$ is $2$-Lipschitz.
	\end{observation}
	
	\begin{sloppypar}
		Let $V$ be the Voronoi diagram of the centers $\{ m_1(M),\dots, m_N(M)\}$ restricted to $\CHull(m_1(M),\dots, m_N(M))$ in the target configuration, as illustrated in Figure~\ref{fig:lowerBoundStartAndTarget}(b).
		For $m \in \{ m_1,\dots,m_N \}$ and $t \in [0,M]$, we denote the Voronoi region of $m(t)$ w.r.t.\, $\{ m_1(t),\dots,m_N(t) \}$ by $V(m(t))$.
		Let $p$ be the trajectory of an arbitrary robot not on the convex hull in the target configuration.
		Furthermore, let $p_1,\dots,p_6 \in \{ m_1,\dots,m_N \}$ be the trajectories of the six robots $1,\dots,6$ adjacent to $p$ in the target configuration.
	\end{sloppypar}
	
	In the following, we show that there is a time interval $I = [t', t' + \frac{1}{20}]$ such that the area of $V(p(t''))$ is lower bounded by $3.479$ for all $t'' \in I$, see Lemma~\ref{lem:largeTriangleInterval}.
	This is larger than the area of $V(p(0))$ and $V(p(M))$ by a constant factor.
	Based on that, we construct the time point $t \in [0,M]$ such that the area of $\CHull(m_1(t),\dots,m_N(t))$ is lower bounded by $cN + \Omega (N^{3/4})$, see Lemma~\ref{lem:lowerBoundArea}.
	To this end, we need to relate the area of a Voronoi region to the length of the corresponding Delaunay edges.

\begin{lemma}
\label{lem:largeArea}
Let $t' \in [0,M]$ and $p(t') \in \{ m_1(t'),\dots,m_N(t') \}$.
	If the maximal distance between $p(t')$ and its Voronoi neighbors is $\lambda \in \big[2, 4 \cos(50^{\circ})\big)$, the area of $V(p(t'))$ is at least
		$$\frac{\lambda}{4} \left( 3 \sin \left( \arccos \left( \frac{\lambda}{4} \right) \right)-  \frac{\lambda}{4} \tan \left( 90^{\circ} - \arccos \left( \frac{\lambda}{4} \right) \right)\right)+\frac{4}{\sqrt{3}},$$ 
	which is at least $3.479$ for $\lambda \in [2.1,2.2]$.
	Furthermore, the area of $V(p(M))$ in the target configuration is $\frac{6}{\sqrt{3}} = 2 \sqrt{3} \leq 3.465$.
\end{lemma}
\begin{proof}
	Let $p_1(t') \in \{ m_1(t'), \dots, m_N(t') \}$ be the center of a robot with $|p(t')p_1(t')| = \lambda$.
	Because all Voronoi neighbors of $p(t')$ have distance less than $4 \cos (50^{\circ})$, the angle between two Voronoi neighbors of $p(t')$ in $p(t')$ is greater than $50^{\circ}$.
	Thus, $p := p(t')$ has at most six Voronoi neighbors and the area of $V(p(t'))$ is minimized if $p(t')$ has five further Voronoi neighbors $p_2(t'),\dots,p_6(t')$.
	We can assume $|p(t')p_2(t')| = \dots = |p(t')p_6(t')| = 2$ because this does not increase the area of $V(p(t'))$.
	W.l.o.g., let $p_1 := p_1(t'),\dots,p_6 := p_6(t')$ be in counterclockwise order around $p$.
	This situation is depicted in Figure~\ref{fig:continuousWorstCaseAreaVoronoiRegionA}.
	
	We find a lower bound on $V(p)$ by lower bounding the intersections of $V(p)$ with the Delaunay triangles that are adjacent to $p$, i.e., with the triangles built by the edges $p_1p_2,\dots,p_5p_6$ and $p_6p_1$ with $p$, see Figure~\ref{fig:continuousWorstCaseAreaVoronoiRegionA}. 
	The area of the two triangles $\triangle_1$ and $\triangle_6$ built by $p_1p_2$ and $p_6p_1$ with $p$ are minimized by assuming the configuration of Figure~\ref{fig:continuousWorstCaseAreaVoronoiRegionA}(a)+(b), i.e., for $|p_1p_2| = |p_1p_6| = 2$. 
	
	\begin{figure}[ht]
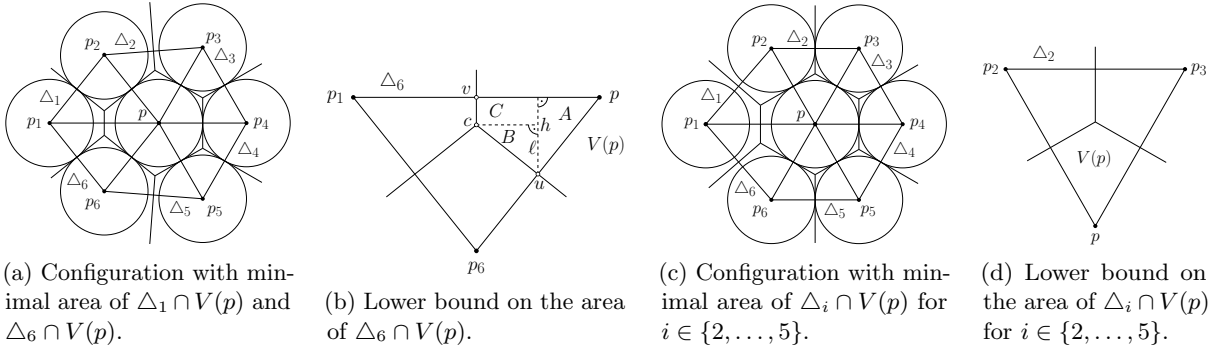

		\begin{center}
			\begin{subfigure}[b]{.23\linewidth}
				\resizebox{\linewidth}{!}{\includegraphics{figures/continuousCaseLowerBoundAreaVRegion3.pdf}}
				\caption{Configuration with minimal area of $\triangle_1 \cap V(p)$ and $\triangle_6 \cap V(p)$.}
			\end{subfigure}
			\hspace{.015\linewidth}
			\begin{subfigure}[b]{.24\linewidth}
				\resizebox{\linewidth}{!}{\includegraphics{figures/continuousCaseLowerBoundAreaVRegion4.pdf}}
				\caption{Lower bound on the area of $\triangle_6 \cap V(p)$.}
			\end{subfigure}
			\hspace{.015\linewidth}
			\begin{subfigure}[b]{.23\linewidth}
				\resizebox{\linewidth}{!}{\includegraphics{figures/continuousCaseLowerBoundAreaVRegion.pdf}}
				\caption{Configuration with minimal area of $\triangle_i \cap V(p)$ for $i \in \{ 2,\dots,5 \}$.}
			\end{subfigure}
			\hspace{.015\linewidth}
			\begin{subfigure}[b]{.18\linewidth}
				\resizebox{\linewidth}{!}{\includegraphics{figures/continuousCaseLowerBoundAreaVRegion2.pdf}}
				\caption{Lower bound on the area of $\triangle_i \cap V(p)$ for $i \in \{ 2,\dots,5 \}$.}
			\end{subfigure}
		\end{center}
		\vspace*{-6pt}
		\caption{Lower bounding the area of $V(p)$ by lower bounding the sum of the areas of the intersections of $V(p)$ with the Delaunay triangles $\triangle_i$ for $i \in \{ 1,\dots,6 \}$ for a maximal distance of~$2.5$ between $p_1$ and $p$.}
		\label{fig:continuousWorstCaseAreaVoronoiRegionA}
	\end{figure}

	In the configuration of Figure~\ref{fig:continuousWorstCaseAreaVoronoiRegionA}(a)+(b), we lower bound the area of $\triangle_6 \cap V(p)$ as follows:
	We subdivide the area of $\triangle_6 \cap V(p)$ into three subsets $A$, $B$, and $C$, and lower bound the area of $\triangle_6 \cap V(p)$ by the sum of lower bounds for $|A|$, $|B|$, and $|C|$.
	
	Let $u$, $v$, and $c$ be the mid points of $pp_6$, $pp_1$, and $\triangle_6$, see Figure~\ref{fig:continuousWorstCaseAreaVoronoiRegionA}(b).
	Furthermore, let $h$ be the vertical side length of $A$ and $\ell$ be the vertical side length of $B$.
	The interior angle of $\triangle_6$ at $p$ is $\arccos (\frac{\lambda}{4})$.
	Thus, we obtain $h = \sin \left( \arccos \left( \frac{\lambda}{4} \right) \right)$, which implies $|A| = \frac{1}{2} \cdot \frac{\lambda}{4} \sin \left( \arccos \left( \frac{\lambda}{4} \right) \right)$.
	The interior angle of $B$ at $c$ is $90^{\circ} - \arccos \left( \frac{\lambda}{4} \right)$.
	Hence, we get $\ell = \frac{\lambda}{4} \tan \left( 90^{\circ} - \arccos \left( \frac{\lambda}{4} \right) \right)$ because the length of $B$'s horizontal side is $\frac{\lambda}{4}$.
	Therefore, $|B| = \frac{1}{2} \cdot \frac{\lambda}{4} \cdot \frac{\lambda}{4} \tan \left( 90^{\circ} - \arccos \left( \frac{\lambda}{4} \right) \right)$. 
	Finally, we have $|C| = \frac{\lambda}{4} (h-\ell) = \frac{\lambda}{4} \left( \sin \left( \arccos \left( \frac{\lambda}{4} \right) \right) - \frac{\lambda}{4} \tan \left(90^{\circ} - \arccos \left( \frac{\lambda}{4} \right) \right)  \right)$. 
	As $\triangle_6 \cap V(p)$ and $\triangle_1 \cap V(p)$ are symmetric, this gives us a lower bound of $2(|A|+|B|+|C|)$ on $|(\triangle_1 \cup \triangle_6) \cap V(p)|$.
	
	Furthermore, the area of $(\triangle_2 \cup \dots \cup \triangle_5) \cap V(p)$ is minimized by the configuration, implying the highest possible packing density, illustrated in Figure~\ref{fig:continuousWorstCaseAreaVoronoiRegionA}(c)+(d) for $|p_2p| = \dots = |p_6p| = |p_2p_3| = \dots = |p_5p_6| = 2$.
	Therefore, this area is at least $4 \cdot \frac{1}{3} \cdot |\triangle_i| = 4 \cdot \frac{1}{3} \cdot \frac{1}{2} \cdot 2 \cdot \sqrt{3} = \frac{4}{\sqrt{3}}$.
	All in all, we obtain $|V(p)| \geq 2(|A| + |B| + |C|) + \frac{4}{\sqrt{3}}$.
	For $\lambda \in [2.1,2.2]$, this is at least $1.17046 + \frac{4}{\sqrt{3}} \geq 3.479$.
	In the target configuration, we have $|p(M)p_i(M)| = 2$ for $i \in \{1,\dots,6\}$. 
	Therefore the area of $V(p(M))$ is $\frac{6}{\sqrt{3}} = 2 \sqrt{3} \leq 3.465$.
 \end{proof}
 
	Next, we prove that there is a time $t'$ with an interval $I :=
[t',t'+\frac{1}{20}]$ during which the area of $\CHull(p,p_1,\dots,p_6)$ is
greater by a constant factor than \textcolor{black}{the area of $\CHull(p,p_1,\dots,p_6)$} in the target configuration. To this end, we
use the following observation that is an immediate consequence of the
intermediate value theorem.

\begin{observation}\label{lem:largeTriangle}
	 There is a time $t' \in [0,M]$ for which the maximal distance between $p(t')$ and $p_1(t'),\dots, p_6(t')$ is $2.2$.
\end{observation}
	

\begin{lemma}\label{lem:largeTriangleInterval}
	There is a time $t' \in [0,M]$ such that for all $t'' \in [t', t' + \frac{1}{20}]$, the area of $V(p(t''))$ is at least $3.479 \geq 1.004 \cdot |V(p(M))|$.
\end{lemma}
\begin{proof}
	Let $\lambda(t)$ be the maximal distance between $p(t)$ and $p_1(t),\dots, p_6(t)$.
	By Observation~\ref{lem:largeTriangle}, there is a maximal time $t'$ with $\lambda(t') = 2.2$.
	Therefore, and because $2.2 < 4 \cos (50^{\circ}) < 2 \sqrt{2}$, the points $p_1(t'),\dots,p_6(t')$ are the Voronoi neighbors of $p(t')$.
	By Observation~\ref{obs:distLipschitz}, $\lambda(t)$ is $2$-Lipschitz.
	This, together with the maximality of $t'$, implies $2.1 \leq \lambda (t'') \leq 2.2$ for $t'' \in [t',t'+\frac{1}{20}]$.
	Thus, Lemma~\ref{lem:largeArea} applies and yields $|V(p(t''))| \geq 3.479 \geq 1.004 \cdot |V(p(M))|$ for all $t'' \in [t',t'+\frac{1}{20}]$.
\end{proof}

\begin{restatable}{lemma}{lemlowerBoundArea}
\label{lem:lowerBoundArea}
	There is a time $t \in [0,M]$ for which the area of $\CHull(m_1(t),\dots,m_N(t))$ is lower-bounded by $3.479 \left( \left\lfloor \frac{N}{20 M} \right\rfloor - \sqrt{2} \pi\left(2\sqrt{N}+M\right)\right) + 2\sqrt{3} \left( N - \left(\left\lfloor \frac{N}{20 M} \right\rfloor \right)\right)$.
\end{restatable}

\begin{proof}
	By Lemma~\ref{lem:largeTriangleInterval}, for each robot $i$ that does not lie on the boundary of the start configuration, there is a point in time $t'' \in [0,M]$ such that the area of $V(p(t''))$ is at least $3.479$ for all $t'' \in [t',t' + \frac{1}{20}]$.
	The continuous pigeonhole principle yields a time point $t \in [0,M]$ such that the area of $k := \left\lfloor \frac{N}{20 M} \right\rfloor \in \Theta (\frac{N}{M})$ Voronoi regions $V(q_1(t)), \dots, V(q_k(t))$ is at least $3.479$.
	For all the remaining Voronoi regions, the area is at least $2 \sqrt{3}$ corresponding to the largest possible packing density as achieved in the start and target configurations.

	We give an upper bound $N \leq \sqrt{2}\pi(2\sqrt{N} + M)$ on the number of robots whose Voronoi regions are not contained in $\CHull(m_1(t), \dots, m_N(t))$.
	W.l.o.g., we assume that all these regions are Voronoi regions whose area we lower bounded by $3.479$.
	Moreover, we can assume all of these regions have zero area\textcolor{black}{, i.e., ignoring when lower bounding the area of $\CHull(m_1(t),\dots,m_N(t))$}.
	Thus, the area of $\CHull(m_1(t),\dots,m_N(t))$ is at least $$\resizebox{.95\textwidth}{!}{$3.479 (k-N) + 2\sqrt{3} (N - k) = 3.479 \left( \left\lfloor \frac{N}{20 M} \right\rfloor - \sqrt{2}\pi\left(2\sqrt{N}+M\right)\right)+ 2\sqrt{3} \left( N - \left(\left\lfloor \frac{N}{20 M} \right\rfloor \right)\right)\text{.}$}$$
	
	\begin{figure}[h]
		\begin{center}
			\resizebox{.9\textwidth}{!}{\includegraphics{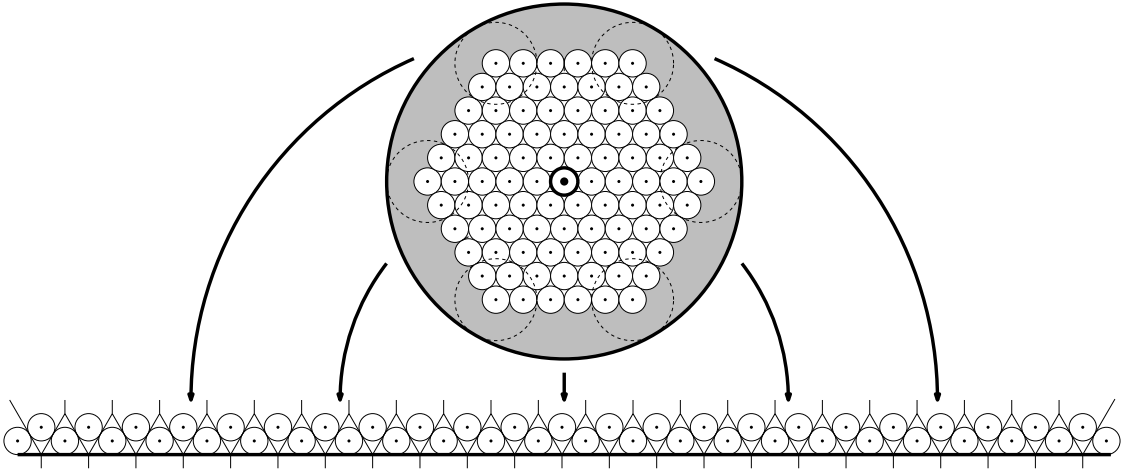}}
		\end{center}
		\caption{An upper-bound construction for the number of robots whose Voronoi regions may intersect the boundary of smallest enclosing ball for $\{ m_1(t),\dots,m_N(t)\}$.
			The radius of the smallest enclosing ball is upper bounded by the distance from the center to the boundary in the start configuration plus the considered makespan illustrated by the dashed circles.}
		\label{fig:upperBoundVRinterBoundaryCH}
	\end{figure}

	It still remains to prove the upper bound $N \leq \sqrt{2} \pi(2\sqrt{N} + M)$ on the number of robots whose Voronoi regions are not contained in $\CHull(m_1(t), \dots, m_N(t))$.
	First, we observe that the length of the boundary of $\CHull(m_1(t), \dots, m_N(t))$ is at most $B := 2\pi(2\sqrt{N} + M)$, because $2\sqrt{N} + M$ is an upper bound on the radius of the smallest ball containing $m_1(t),\dots,m_N(t)$.
	In order to estimate $N$, we consider the maximal number of points from $[0, B] \times \mathbb{R}_{\geq 0}$ whose Voronoi regions intersect the $x$-axis.
	This number is achieved for the configuration as illustrated in Figure~\ref{fig:upperBoundVRinterBoundaryCH}, implying $N \leq \frac{B}{\sqrt{2}} = \sqrt{2}\pi(2\sqrt{N} + M)$, thus concluding the proof.
\end{proof}

\begin{lemma}\label{lem:continuousCaseUpperBoundArea}
	For each $t \in [0,M]$, $|\CHull(m_1(t),\dots,m_N(t))| \leq 2 \sqrt{3}N+ 2\pi(\sqrt{N} + M) M$.
\end{lemma}
\begin{sloppypar}
\begin{proof}
	In the start configuration, the intersection of each Voronoi cell with $\CHull(m_1(0),\dots,m_N(0))$ has an area of $2 \sqrt{3}$.
	Thus, the convex hull of the start configuration has an area of at most $2 \sqrt{3}N$.
	We give an upper bound on the area $A$ gained during the motion, i.e., the area of $\CHull(m_1(t),\dots,m_N(t)) \setminus \CHull(m_1(0),\dots,m_N(0))$, corresponding to the gray region in Figure~\ref{fig:upperBoundVRinterBoundaryCH}.
	The length of the boundary $\partial\CHull(m_1(t),\dots,m_N(t))$ is at most $2\pi(\sqrt{N} + M)$, implying $A \leq 2\pi(\sqrt{N} + M) M$, thus concluding the proof.
\end{proof}
\end{sloppypar}

\statement{Theorem}{thm:lowerBoundContinuousMakespan}
	{\em
	There is an instance with optimal makespan $M \in \Omega (N^{1/4})$, see Figure~\ref{fig:lowerBoundStartAndTarget}.
}
\begin{proof}
	Combining the bounds from Lemma~\ref{lem:lowerBoundArea} and Lemma~\ref{lem:continuousCaseUpperBoundArea} yields black the following 
		\begin{eqnarray*}
			&&3.479 \left( \left\lfloor \frac{N}{20 M} \right\rfloor - \sqrt{2}\pi\left(2\sqrt{N}+M\right)\right)+ 2\sqrt{3} \left( N - \left(\left\lfloor \frac{N}{20 M} \right\rfloor \right)\right)\\
			&\leq &2 \sqrt{3}N+ 2\pi(\sqrt{N} + M) M.\\
			\Leftrightarrow && 0,014  \left\lfloor \frac{N}{20 M} \right\rfloor - 3,479\sqrt{2}\pi\left(2\sqrt{N}+M\right)\leq 2\pi(\sqrt{N} + M) M.
		\end{eqnarray*}
	If $M \in \Omega(\sqrt{N})$ holds, we are done.
	Otherwise we \textcolor{black}{obtain} $\frac{N}{M} \in \mathcal{O}(M \sqrt{N}) \Leftrightarrow \sqrt{N} \in \mathcal{O}(M^2)$, and thus $M \in \Omega (N^{1/4})$, concluding the proof.
\end{proof}

\subsection{An Upper Bound for Unbounded Environments}
\label{subsec:upper_cont}
Next we give upper bounds on the stretch and makespan for moving disks in unbounded environments.
First, we show that we can achieve constant stretch for well-separated robots.

\begin{theorem}
\label{th:cont}
If the distance between the centers of two robots of radius $1$ is at least $4$ in the start and target configurations, we can achieve a makespan in $\mathcal{O}(d)$, i.e., constant stretch.
\end{theorem}

\begin{figure}[h]
	\begin{center}
		\includegraphics[scale=0.5]{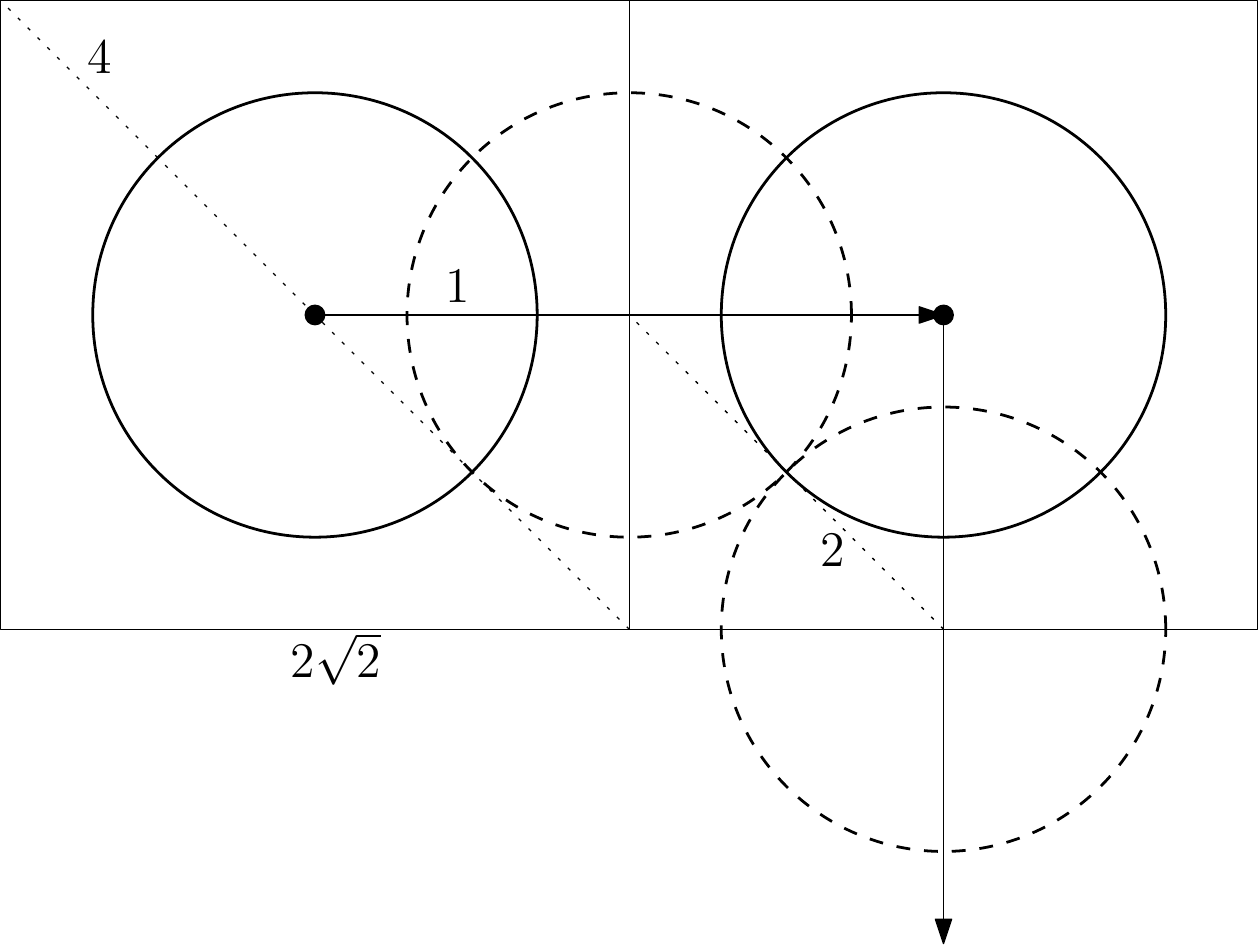}
	\end{center}
	\vspace*{-6pt}
	\caption{A mesh size of $2 \sqrt{2}$ avoids robot collisions, and the cell diagonals have length $4$. Note that robots may have arbitrary shape, as the separation argument applies to their circumcircles.}
	\label{fig:continuousUpperBoundMeshSize}
	\vspace*{-6pt}
\end{figure}

\begin{proof}
	We consider a grid $D$ with mesh size $2\sqrt{2}$.
	In this way, as shown in Figure~\ref{fig:continuousUpperBoundMeshSize}, two robots starting simultaneously from different cells and traveling along two incident edges can touch when they reach the midpoints, but do not collide.
	Moreover, the diagonals have length $4$.
	By choosing a grid that has no robot center on a grid line, 
every cell of $D$ contains at most \textcolor{black}{one} start and one target position of a robot.
	Additionally, we can move each robot in the start and target configuration to the center of its own cell, allowing us to 
	to use our algorithm from Section~\ref{sec:constappr}. Overall, we
achieve a set of trajectories with makespan in $\mathcal{O}(d)$. 
\end{proof}

In the remainder of this section we give an $\mathcal{O}(\sqrt{N})$-approximation algorithm 
for the continuous makespan for these kind of {\em well-separated} arrangements,
by extending the approach for discrete grids.

Again, we make use of an underlying grid with mesh size $2\sqrt{2}$.
Our algorithm proceeds in three phases.
(1) Moving the robots to vertices of the grid,
(2) applying our $\mathcal{O}(1)$-approximation for the discrete case, and
(3) moving the robots from the vertices of the grid to their target positions.
To ensure a $\mathcal{O}(\sqrt{N})$-approximation, we move each robot center to a grid vertex within a distance of $\mathcal{O}(\sqrt{N})$.
Phase\textcolor{black}{s} (1) and (3) are symmetric in the following sense.
By applying the steps of phase (1) to the target configuration in reverse, we compute a grid configuration that serves as target configuration for phase (2).

Phase (1) works as follows.
(1.1) We begin by sorting the $N$ robots according to the $(x,y)$-lexicographical order.
Then we subdivide them into $\left\lceil \sqrt{N} \right\rceil$ vertical slices, each containing at most $\left\lceil \sqrt{N} \right\rceil$ robots.
To the right of every slice, we add a vertical buffer slice of width $4\sqrt{2}$ by moving all robots not yet considered \textcolor{black}{by} $4\sqrt{2}$ units to the right.
These trajectories are used in parallel, the distance covered by each robot is in $\mathcal{O}(\sqrt{N})$.
The buffer slices guarantee that in all following steps, the robots in each vertical slice are independent of each other.

For (1.2), we continue by sorting the robots within the vertical slices according to the $(y,x)$-lexicographical order.
We separate the robots by ensuring vertical distance at least $4\sqrt{2}$ between every pair of robots.
This can be done by moving the robots upwards, starting from the second-to-lowest one.
These trajectories can be done in parallel and the distance covered by each robot is in $\mathcal{O}(\sqrt{N})$.
For (1.3), we finally move each robot to the bottom-left vertex of the grid cell containing its center.

\statement{Theorem}{thm:continuousCaseUpperBound}
	{\em
	The algorithm described above computes a trajectory plan with continuous makespan in $\mathcal{O}(d+\sqrt{N})$.
	If $d \in \Omega(1)$, this implies a $\mathcal{O}(\sqrt{N})$-approximation algorithm.
}
\begin{proof}
	Phase~(1) guarantees that either the horizontal or the vertical distance between each pair of robots is at least $4\sqrt{2}$.
	Therefore, in each grid cell, there is at most one robot and each robot is moved to its own grid vertex.
	In phase~(2), each robot is moved by $\mathcal{O}(d')$ units, where $d'$ is the maximal distance between a robot's start and target position in the grid.
	As the distance each robot covers in phases~(1)~and~(3) is $\mathcal{O}(\sqrt{N})$, the distance traveled in phase~(2) is in $\mathcal{O}(d+\sqrt{N})$.
	Therefore, the trajectory set computed by the algorithm has continuous makespan in $\mathcal{O}(d+\sqrt{N})$.
	The running time as described above is pseudopolynomial; it becomes polynomial by using standard compression techniques, e.g., by compressing large empty rectangles.
\end{proof}

\subsection{Colored and Unlabeled Disks}

We can combine the positive results of the previous section with the technique of Theorem~\ref{th:cont}
to achieve the same result for colored (and in particular, unlabeled) disks.

\begin{corollary}
\label{co:cont_unlabeled}
There is an algorithm with running time $\mathcal{O}(k(mn)^{1.5} \log (mn) +
dmn)$ that computes, given start and target images $I_s,I_t$, an
$\mathcal{O}(1)$-approximation of the optimal makespan $M$ and a corresponding
set of trajectories.  
\end{corollary}

\begin{proof}
The proof proceeds analogous to Theorem~\ref{thm:unlabeled}: After computing an optimal bottleneck matching,
apply Theorem~\ref{thm:main} in the setting of Theorem~\ref{th:cont}.
\end{proof}

\section{Geometric Difficulties}
\label{sec:difficulties}

From a practical point of view, it is also desirable to compute 
provably optimal trajectories for specific instances of
moderate size, instead of solutions for large instances
that are within a provable constant factor of the optimum.
This also plays a role as a building block for other 
purposes, e.g., providing a formal proof of NP-hardness
for parallel geometric motion planning for disks: We need to be able to
establish the shape of optimal trajectories when building gadgets
for a hardness construction.

\begin{figure}[h]
                \begin{center}
		\resizebox{.45\textwidth}{!}{\includegraphics{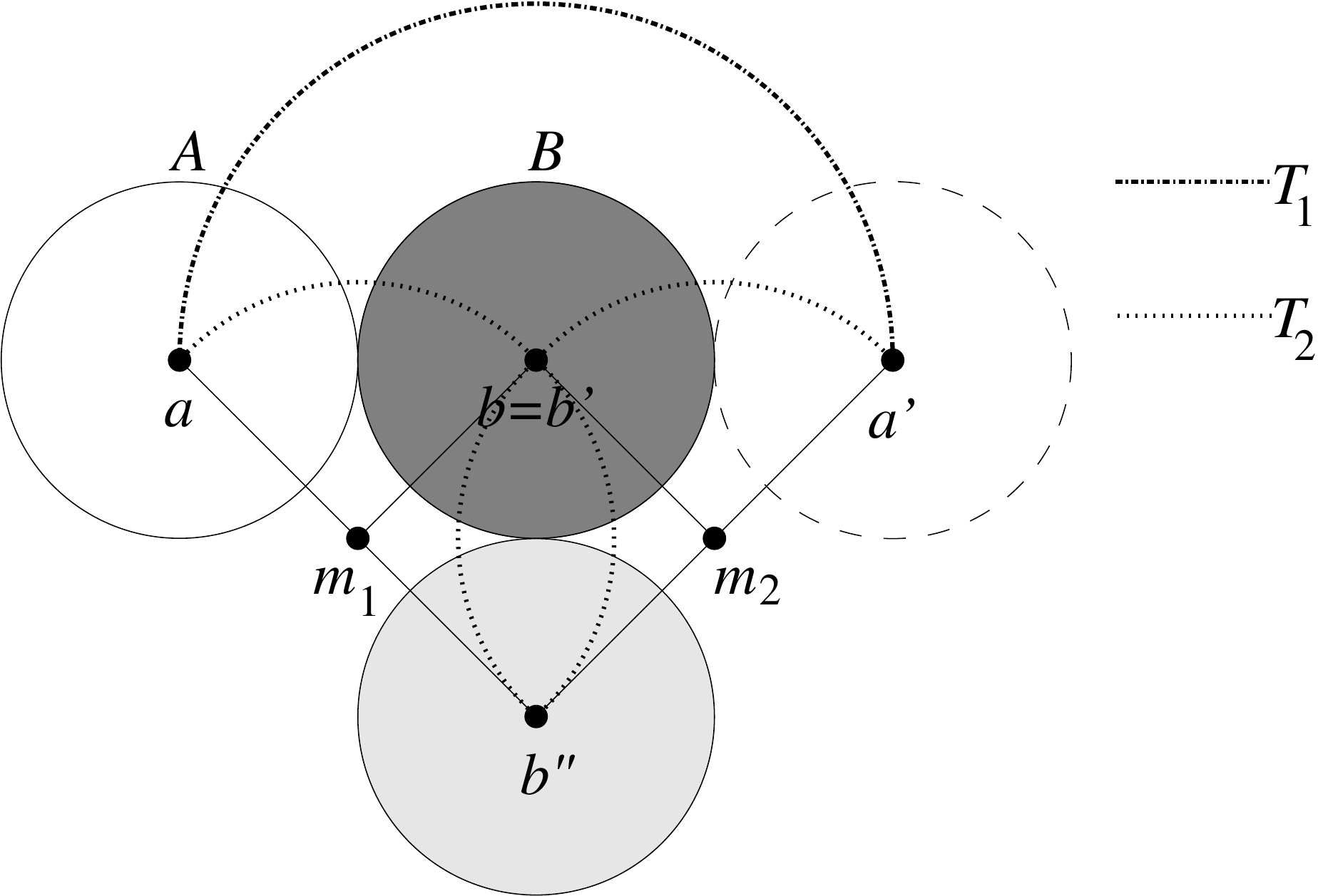}}                
		\resizebox{.45\textwidth}{!}{\includegraphics{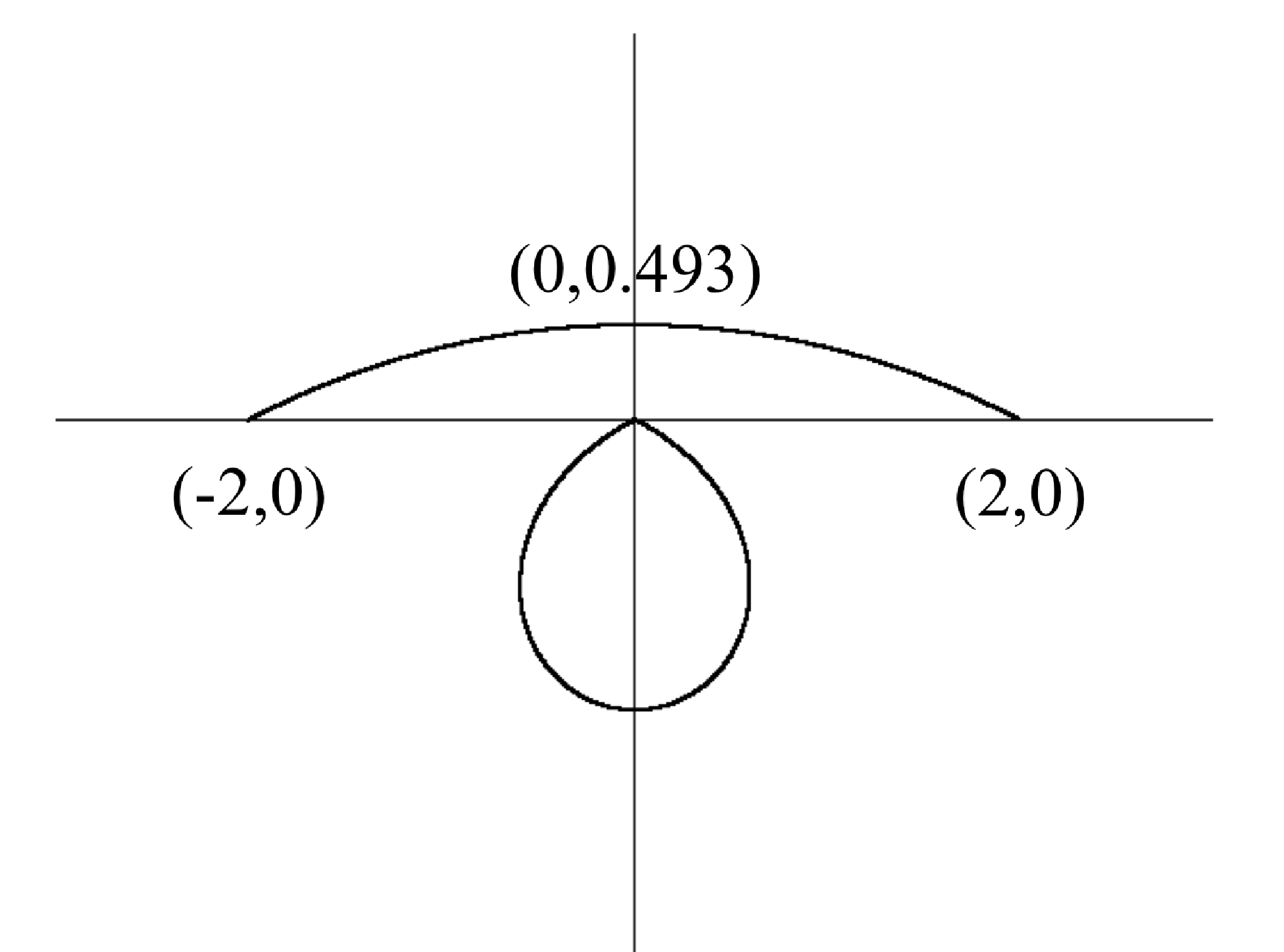}}                
                \end{center}
  \vspace*{-12pt}
                \caption{Moving the left unit disk $A$ from position $a$ to position $a'$ at distance 4, with 
disk $B$ starting and ending at $b=b'$. (Left)
Trajectory $T_1$ rotates disk $A$ around the stationary disk $B$ resulting in makespan $2\pi=6.28\ldots$.
Trajectories $T_2$ rotate both disks around the centers $m_1$ and $m_2$, resulting in makespan 
$\sqrt{2}\pi=4.44\ldots$. (Right) Choosing a circular arc through $(-2,0), (2,0)$ and the (numerically optimized) 
point $(0,0.493\ldots)$ for disk $A$ (with $B$ moving accordingly at distance 2) yields the trajectory $T_3$ with makespan $4.16\ldots$.}
                \label{fig:passing}
  \vspace*{-12pt}
        \end{figure}

However, the geometry involved in this goal is far from easy, even for an example
as small as the one shown in Figure~\ref{fig:passing}.
\revised{This is closely related to recent work by Kirkpatrick and Liu~\cite{kl-cmlcm2d-16}, who devote a whole paper 
to computing optimal trajectories for two disks in an arbitrary initial and target configuration, with the objective of minimizing
the total distance travelled instead of the makespan. A key insight is that optimal trajectories consist of a limited number of
circular arcs. This is not necessarily the case for trajectories that minimize the makespan.}
Even for the seeming simplicity of our example, we do not have a proof of optimality of the trajectory $T_3$ shown
on the right. This illustrates the difficulty of characterizing and establishing optimal trajectories that minimize the total
duration of a parallel schedule, highlighting the special role of geometry for the problem.

\section{Conclusion}
\label{sec:concl}
We have presented progress on a number of algorithmic problems of parallel motion
planning, also shedding light on a wide range of interesting open problems \textcolor{black}{described in the following}.

The first set of problems consider complexity.
The labeled problem of Section~\ref{sec:constappr} is known to be NP-complete
for planar graphs. 
It is natural to conjecture that the geometric version
is also hard. 
It seems tougher 
to characterize the family of optimal trajectories:
As shown above, 
their nature is unclear, so membership in NP is doubtful.

A second set of questions considers the relationship between
stretch factor and disk separability in the continuous setting.
We believe that the upper bound of $\mathcal{O}(\sqrt{N})$ on the worst-case stretch factor for dense arrangements 
is tight. 
What is the critical separability of disks
for which constant stretch can be achieved?
How does the stretch factor increase as a function of $N$
below this threshold?
For {\em sparse} arrangements
of disks, simple greedy, straight-line trajectories between the origins
and destinations \textcolor{black}{of disks} encounter only isolated conflicts,
resulting in small stretch factors close to 1, i.e., $1+o(1)$. 
What is the relationship between (local) density and the achievable stretch factor
along the whole density spectrum?

Finally, practical motion planning requires a better handle on
characterizing and computing optimal solutions for specific instances, along with
lower bounds, possibly based on numerical methods and tools. 
Moreover, there is a wide range of additional objectives and requirements,
such as accounting for acceleration or deceleration of disks, turn cost,
or multi-stop tour planning. All these are left for future work.

\subparagraph*{Acknowledgements.}

\nnew{We thank anonymous reviewers of a preliminary version of the paper for helping to improve the overall presentation.}

\bibliography{refs}



\end{document}